\documentclass[11pt]{article}
\usepackage{lp,graphicx,etoolbox,enumitem}
\usepackage{amsmath, amssymb}
\usepackage{subfig}

\usepackage[noline,noend,linesnumberedhidden,boxed]{algorithm2e}
\DontPrintSemicolon
\SetArgSty{textnormal}
\SetAlCapSkip{1ex}
\SetAlgoCaptionLayout{centerline}

\newcommand{\qed}{\rule{7pt}{7pt}}

\newenvironment{proof}{\noindent {\it Proof\/}:}{$\qed$ \medskip}
\newcommand{\set}[1]{\{ #1 \}}
\newtheorem{theorem}{Theorem}[section]
\newtheorem{lemma}[theorem]{Lemma}
\newtheorem{corollary}[theorem]{Corollary}
\newtheorem{conjecture}{Conjecture}

\newtheorem{claim}[theorem]{Claim}

\def\({\left(}
\def\){\right)}

\newlength{\fixSUBT}
\newcommand{\SUBT}{\ensuremath{SU\hspace*{-\fixSUBT}BT}}

\newtoggle{abs}
\togglefalse{abs}

\begin{document}

\title{A Proof of the Boyd-Carr Conjecture}
\author{Frans Schalekamp
\and David P.\ Williamson\thanks{Address: School
of Operations Research and Information Engineering, Cornell
University, Ithaca, NY 14853, USA.  Email: {\tt dpw@cs.cornell.edu}.  This work was carried out while the author was on sabbatical at TU Berlin.
Supported in part by the Berlin Mathematical School, the Alexander von Humboldt Foundation, and NSF grant CCF-1115256.}
\and Anke van Zuylen\thanks{Address: Max-Planck-Institut f\"ur Informatik, Department 1: Algorithms and Complexity, Campus E1 4, Room 311c, 66123 Saarbr\"ucken, Germany. Email: {\tt anke@mpi-inf.mpg.de}.}}

%\date{}
\maketitle
\thispagestyle{empty}

\begin{abstract}
Determining the precise integrality gap for the subtour LP relaxation of the traveling salesman problem is a significant open question, with little progress made in thirty years in the general case of symmetric costs that obey triangle inequality.  Boyd and Carr~\cite{BoydC11} observe that we do not even know the worst-case upper bound on the ratio of the optimal 2-matching to the subtour LP; they conjecture the ratio is at most 10/9.

In this paper, we prove the Boyd-Carr conjecture.  In the case that a fractional 2-matching has no cut edge, we can further prove that an optimal 2-matching is at most 10/9 times the cost of the fractional 2-matching.
\end{abstract}

\setcounter{page}{0}
\newpage

\section{Introduction}

The traveling salesman problem (TSP) is the most famous problem in discrete optimization.  Given a set of $n$ cities and the costs $c(i,j)$ of traveling from city $i$ to city $j$  for all $i,j$, the goal of the problem is to find the least expensive tour that visits each city exactly once and returns to its starting point.  An instance of the TSP is called {\em symmetric} if $c(i,j) = c(j,i)$ for all $i,j$; it is {\em asymmetric} otherwise.  Costs obey the {\em triangle inequality} if $c(i,j) \leq c(i,k) + c(k,j)$ for all $i,j,k$.  The TSP is known to be NP-hard, even in the case that instances are symmetric and obey the triangle inequality.  From now on we consider only these instances unless otherwise stated.

Because of the NP-hardness of the traveling salesman problem, researchers have considered approximation algorithms for the problem.  The best approximation algorithm currently known is a $\frac 32$-approximation algorithm given by Christofides in 1976~\cite{Christofides76}.  Better approximation algorithms are known for special cases.  Exciting progress has been made recently in the case of the graphical TSP, in which costs $c(i,j)$ are given by shortest path distances in an unweighted graph; M\"omke and Svensson~\cite{MomkeS11} give a 1.461-approximation algorithm for this case.  However, to date, Christofides' algorithm has the best known performance guarantee for the general case.

There is a well-known, natural direction
for making progress which has also defied
improvement for nearly thirty years.  The following linear programming
relaxation of the traveling salesman problem was used by Dantzig,
Fulkerson, and Johnson~\cite{DantzigFJ54} in 1954.  For simplicity of notation, we
let $G=(V,E)$ be a complete undirected graph on $n$ vertices.  In
the LP relaxation, we have a variable $x(e)$ for all $e = (i,j)$
that denotes whether we travel directly between cities $i$ and $j$
on our tour.  Let $c(e) = c(i,j)$, and let $\delta(S)$ denote the
set of all edges with exactly one endpoint in $S \subseteq V$.
Then the relaxation is
\lps & &
& \mbox{Min} & \sum_{e \in E} c(e) x(e) \\
(\SUBT) & \mbox{subject to:} & & & \sum_{e \in \delta(i)} x(e) = 2, & \forall i \in V, \numb{degreecons} \\
& & & & \sum_{e \in \delta(S)}x(e) \geq 2, & \forall S\subset V,\, 3 \leq |S| \leq |V|-3 \numb{subtourcons}\\
& & & & 0 \leq x(e) \leq 1, & \forall e \in E. \numb{boundscons} \elps
The first
set of constraints (\ref{degreecons}) are called the {\em degree constraints}.  The second set
of constraints (\ref{subtourcons}) are sometimes called {\em
subtour elimination constraints} or sometimes just {\em subtour
constraints}, since they prevent solutions in which there is a
subtour of just the vertices in $S$. As a result, the linear program is sometimes called the {\em subtour LP}.
\iftoggle{abs}{}{
It is known that
the equality sign in the first set of constraints may be replaced by $\ge$
in case the costs obey the
triangle inequality (Goemans and Bertsimas~\cite{GoemansB90}; see also Williamson~\cite{Williamson90}).}

The LP is known to give excellent lower bounds on TSP instances in practice, coming within a percent or two of the length of the optimal tour
(see, for instance, Johnson and McGeoch~\cite{JohnsonM02}).  However, its theoretical worst-case is not well understood.  In 1980, Wolsey~\cite{Wolsey80} showed that Christofides'
algorithm produces a solution whose value is at most $\frac 32$
times the value of the subtour LP (also shown later by Shmoys and
Williamson~\cite{ShmoysW90}). This proves
that the {\em integrality gap} of the subtour LP is at most
$\frac 32$; the integrality gap is the worst-case ratio, taken
over all instances of the problem, of the value of the optimal tour to the value of the subtour LP, or the ratio of the optimal integer solution to the optimal fractional solution.   The integrality gap of the LP is
known to be at least $\frac 43$ via a specific class of instances.  However, no instance is known that has integrality gap worse than this, and it has been conjectured for some time that the integrality gap is at most $\frac 43$ (see, for instance, Goemans~\cite{Goemans95}).   The results of M\"omke and Svensson~\cite{MomkeS11} show that in the case of the graphical TSP, the integrality gap is at most 1.461; if the graph is cubic, Boyd, Sitters, van der Ster, and Stougie~\cite{BoydSSS11} show that the gap is $\frac 43$, and M\"omke and Svensson extend this bound to subcubic graphs as well.

\iftoggle{abs}{}{
There is some evidence that the conjecture might be true.
Benoit and Boyd~\cite{BenoitB08} have shown via computational
methods that the conjecture holds for $n \leq 10$, and Boyd and Elliot-Magwood~\cite{BoydE10} have extended this to $n \leq 12$.  In a 1995 paper,
Goemans~\cite{Goemans95} showed that adding any class of valid
inequalities known at the time to the subtour LP could increase
the value of the LP by at most $\frac 43$; this is necessary
for the conjecture to be true.  Somewhat weaker evidence is as follows.
A {\em 2-matching} is an integer solution to the subtour LP obeying only the degree constraints (\ref{degreecons}) and the bounds constraints (\ref{boundscons}).\footnote{We note that what we refer to here as 2-matchings, are also sometimes called 2-factors.}
A {\em fractional 2-matching} is a 2-matching without the integrality constraints.
Boyd and Carr~\cite{BoydC99} have shown that the integrality gap for the
2-matching problem is at most $\frac 43$. Furthermore,  Boyd and Carr~\cite{BoydC11} have shown that if the subtour LP solution is half-integral (that is, $x(i,j) \in \set{0,\frac{1}{2},1}$ for all $i,j \in V$) and has a particular structure then there is a tour of cost at most $\frac 43$ times the value of the subtour LP.
}

Not only do we not know the integrality gap of the subtour LP, Boyd and Carr have observed that we don't even know the worst-case ratio of the optimal 2-matching to the value of the subtour LP, which is surprising because 2-matchings are well understood and well characterized.
\iftoggle{abs}{A {\em 2-matching} is an integer solution to the subtour LP obeying only the degree constraints (\ref{degreecons}) and the bounds constraints (\ref{boundscons}).\footnote{We note that what we refer to here as 2-matchings are also sometimes called 2-factors.}
A {\em fractional 2-matching} is a 2-matching without the integrality constraints.  Boyd and Carr make the following conjecture.}{They make the following conjecture.}
\begin{conjecture}[Boyd and Carr~\cite{BoydC11}] \label{conj:bc}
The worst-case ratio of an optimal 2-matching to an optimal solution to the subtour LP is at most $\frac {10}9$.\iftoggle{abs}{\footnote{The especially astute reader may wonder why this conjecture does not follow from the work of Goemans \cite{Goemans95}.  See Boyd and Carr \cite{BoydC11} or the full version of this paper for a discussion of this issue.}}{}
\end{conjecture}
\noindent It is known that there are cases for which the cost of an optimal 2-matching is at least $\frac{10}9$ times the optimal solution to the subtour LP; see Figure~\ref{fig:worstcase}.  Boyd and Carr have shown that the conjecture is true if the solution to the subtour LP has a very special structure: namely, all variables $x(e) \in \set{0,\frac 12,1}$, the cycles formed by the edges $e$ with $x(e) = \frac 12$ all have the same odd size $k$, and the support is $(k-1)$-edge-connected.\footnote{In fact, they show in this case the optimal 2-matching has cost at most $\frac{3k+1}{3k}$ times the subtour LP.}
\iftoggle{abs}{In the general case, the only bound on this ratio we know is one of Boyd and Carr \cite{BoydC99}, who show that the integrality gap of 2-matchings is at most $\frac 43$}{In the general case, the only bound on this ratio that we know of is the Boyd and Carr bound on the integrality gap of 2-matchings}; since the constraints of the subtour LP are a superset of the fractional 2-matching constraints, this implies the ratio is at most~$\frac 43$.

\begin{figure}
\begin{center}
\includegraphics[height=1in]{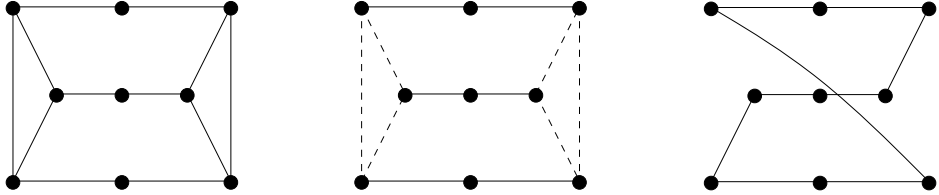}
\end{center}
\caption{Illustration of the worst example known for the ratio of 2-matchings to the subtour LP.  The figure on the
left shows the instance; all edges in the graph have cost 1, all other edges have cost 2.  The figure in the center gives the subtour LP
solution, in which the dotted edges have value $\frac 12$, and the solid
edges have value 1; this is also an optimal fractional 2-matching.  The figure on the right gives an optimal 2-matching, which is also the optimal tour.} \label{fig:worstcase}
\end{figure}

\iftoggle{abs}{}{
The work of Goemans~\cite{Goemans95} has some bearing on this conjecture.  He studies the following linear program which is essentially same as the subtour LP in the case edge costs obey triangle inequality:
\lps & &
& \mbox{Min} & \sum_{e \in E} c(e) x(e) \\
(\SUBT') &  \mbox{subject to:} & & & \sum_{e \in \delta(S)} x(e) \geq 2, & \forall S \subset V,\, S \neq \emptyset, \numb{subtourcons2} \\
& & & &  x(e) \geq 0, & \forall e \in E. \numb{boundscons2} \elps
Goemans shows (among other things) that adding comb inequalities to this LP can increase the LP value by at most $\frac{10}9$; more precisely, he shows that if $x$ is a feasible solution to $(\SUBT')$, then $\frac{10}{9} x$ is feasible for the LP obtained by adding comb inequalities to $(\SUBT')$.  It is known that adding a subset of the comb inequalities to the degree constraints (\ref{degreecons}) and bounds (\ref{boundscons}) gives the 2-matching polytope.  This would imply the Boyd-Carr conjecture if it were known that there is an optimal solution that obeys the degree constraints when the comb inequalities are added to $(\SUBT')$; as mentioned above, it can be shown that there is an optimal solution  for $(\SUBT')$ that obeys the degree constraints when the edge costs obey the triangle inequality.  But we do not know whether there is an optimal solution that obeys the degree constraints if the comb inequalities are added.\footnote{To quote Goemans \cite[p.\ 348]{Goemans95}: ``One might wonder whether the worst-case improvements remain unchanged when one adds the degree constraints $x(\delta\{i\})=2$ for all $i \in V$ and restricts one's attention to cost functions satisfying the triangle inequality.  We believe so but have been unable to prove it.  The result would follow immediately if one could prove that the degree constraints never affect the value of the relaxation when the cost function satisfies the triangle inequality.''}}

The contribution of this paper is to improve our state of knowledge for the subtour LP by proving Conjecture~\ref{conj:bc}.  

We start by showing that in some cases the cost of an optimal 2-matching is at most $\frac{10}9$ the cost of a fractional 2-matching, which is a stronger statement than Conjecture~\ref{conj:bc}; in particular, we show this is true whenever the support of the fractional 2-matching has no cut edge.  The example in Figure~\ref{fig:worstcase} shows that the ratio can be at least $\frac{10}9$ in such cases, so this result is tight.  As the first step in this proof, we give a simplification of the Boyd and Carr result bounding the integrality gap for 2-matchings by~$\frac 43$.  In the case that the support of an optimal fractional 2-matching has no cut edge, the proof becomes quite simple.  
The perfect matching polytope plays a crucial role in the proof: we use the matching edges to show us which edges to remove from the solution in addition to showing us which edges to add. We note that this idea was independently developed in the recent work of M\"omke and Svensson, but also previously appeared in the reduction of the 2-matching polytope to the matching polytope; see, for instance, Schrijver~\cite[Section 30.7]{Schrijver-book}.
We also use a notion from Boyd and Carr~\cite{BoydC99} of a {\em graphical} 2-matching: in a graphical 2-matching, each vertex has degree either 2 or 4, each edge has 0, 1, or 2 copies, and each component has size at least three.  Given the triangle inequality, we can shortcut any graphical 2-matching to a 2-matching of no greater cost.

To obtain our proof of the Boyd-Carr conjecture, we give a polyhedral formulation of the graphical 2-matching problem, and use it to prove Conjecture~\ref{conj:bc}.  If $x$ is a feasible solution for the subtour LP, then, roughly speaking, we show that $\frac{10}{9} x$ is feasible for the graphical 2-matching polytope.  Our previous results give us intuition for the precise mapping of variables that we need.  Using the graphical 2-matching polytope allows us to overcome the issues with the degree constraints faced in trying to use Goemans' results.

All the results above can be made algorithmic and have polynomial-time algorithms, though we do not explicitly determine running times.

We conclude by posing a new conjecture, namely that the worst-case integrality gap is achieved for solutions to the subtour LP that are fractional 2-matchings (that is, for instances such that adding the subtour constraints to the degree constraints and the bounds on the variables does not change the objective function value).

In a companion paper, Qian, Schalekamp, Williamson, and van Zuylen~\cite{QianSWvZ11} show that the proof of the Boyd-Carr conjecture can be used to help bound the integrality gap of the subtour LP for the 1,2-TSP.  They show that the gap is at most $\frac{106}{81} \approx 1.3086 < \frac 43$.  They also give a proof that the cost of the optimal 2-matching is at most $\frac{10}9$ times the cost of a fractional 2-matching in the case that $c(i,j) \in \{1,2\}$, which gives an alternate proof of the Boyd-Carr conjecture in this case.

Our paper is structured as follows.  We introduce basic terms and notation in Section~\ref{sec:prelims}.  In Section~\ref{sec:2m-comb}, we rederive the Boyd-Carr integrality gap for 2-matchings, and show that the gap is at most $\frac{10}9$ in the case the fractional 2-matching has no cut edge. In Section~\ref{sec:2m-poly}, we give the polytope for graphical 2-matchings and show how to use it to prove the Boyd-Carr conjecture. Finally, we close with our new conjecture in Section~\ref{sec:conc}.  
\iftoggle{abs}{Some proofs are omitted due to space restrictions.}{}

\section{Preliminaries}
\label{sec:prelims}

We will work extensively with fractional 2-matchings; that is, optimal solutions $x$ to the LP $(\SUBT)$ with only constraints (\ref{degreecons}) and (\ref{boundscons}).  For convenience we will abbreviate ``fractional 2-matching'' by F2M and ``2-matching'' by 2M.  F2Ms have the following well-known structure (attributed to Balinski~\cite{Balinski65}).  Each connected component of the support graph (that is, the edges $e$ for which $x(e) > 0$) is either a cycle on at least three vertices with $x(e)=1$ for all edges $e$ in the cycle, or consists of odd-sized cycles with $x(e) = \frac 12$ for all edges $e$ in the cycle connected by paths of edges $e$ with $x(e) = 1$ for each edge $e$ in the path (the center figure in Figure~\ref{fig:worstcase} is an example).  We call the former components {\em integer components} and the latter {\em fractional components}.  Many of our results focus on transforming an F2M into a 2M, in which all components are integer.  For that reason, we will often focus solely on how to transform the fractional components into integer components. We then call the edges of fractional components for which $x(e)=\frac 12$ {\em cycle edges} and the edges for which $x(e)=1$ {\em path edges}.
Note that removing a cycle edge can never disconnect a fractional component.  If removing a path edge disconnects a fractional component, we call it a {\em cut edge}.  The associated path of the path edge we will call a {\em cut path}, since every edge in it will be a cut edge.  We will say that a fractional 2-matching is {\em connected} if it has a single component.

We will use a concept introduced by Boyd and Carr~\cite{BoydC99} of a {\em graphical} 2-matching (G2M).  As stated above, in a graphical 2-matching, each vertex has degree either 2 or 4, each edge has 0, 1, or 2 copies, and each component has size at least three.  Given the triangle inequality, we can shortcut any G2M to a 2M of no greater cost.  Our techniques for transforming an F2M to a 2M actually find G2Ms.

We will often need to find minimum-cost perfect matchings.  By a result of Edmonds~\cite{Edmonds65}, the perfect matching polytope is defined by the following linear program $(M)$:
\lps & &
& \mbox{Min} & \sum_{e \in E} c(e) x(e) \\
(M) & \mbox{subject to:} & & & \sum_{e \in \delta(i)} x(e) = 1, & \forall i \in V,  \numb{matchingdegreecons} \\
& & & & \sum_{e \in \delta(S)}x(e) \geq 1, & \forall S\subset V,\, |S| \mbox{ odd}, \numb{matchingScons}\\
& & & & x(e) \geq 0, & \forall e \in E. \numb{matchingboundscons} \elps

\section{2-matching Integrality Gaps}
\label{sec:2m-comb}

In this section, we bound the cost of a G2M in terms of an F2M via combinatorial methods.  We start by giving a proof of a result of Boyd and Carr~\cite{BoydC99} that there is a G2M of cost at most $\frac{4}{3}$ the cost of an F2M.    Our proof is somewhat simpler than theirs, but more importantly, it introduces the main ideas that we will need to obtain other results.  We then show that if the F2M has no cut edges, we can improve the bound from $\frac{4}{3}$ to $\frac{10}{9}$.  The main idea of this section is that given an F2M, we define a matching problem and compute a perfect matching.  The perfect matching tells us how to modify the fractional components by either duplicating or removing edges so that we obtain a G2M.  We then relate the cost of the perfect matching found to the F2M by providing a feasible solution to the perfect matching LP $(M)$.  We will need the following result of Naddef and Pulleyblank~\cite{NaddefP81}
\iftoggle{abs}{.}{; we give the proof since we will use some of its ideas later on.}

\begin{lemma}[Naddef and Pulleyblank~\cite{NaddefP81}]
\label{lem:np}
Let $G$ be a cubic, 2-edge-connected graph with edge costs $c(e)$ for all $e \in E$.  Then there exists a perfect matching in $G$ of cost at most $\frac{1}{3} \sum_{e \in E} c(e)$.
\end{lemma}

\iftoggle{abs}{}{
\begin{proof}
The main idea is to show that $x(e) = \frac 13$ is a feasible solution to the matching polytope $(M)$.  The lemma then follows from the fact that $(M)$ has integer extreme points. Since $G$ is cubic, $|V|$ must be even, and $\sum_{e \in \delta(i)} x(e) = 1$.  Now consider any $S \subset V$ with $|S|$ odd.  Because $G$ is cubic, it must be that $|\delta(S)|$ is odd, and since $G$ is 2-edge-connected, $|\delta(S)| \geq 2$.  Therefore $|\delta(S)| \geq 3$, and $\sum_{e \in \delta(S)} x(e) \geq 1$.
\end{proof}
}

\begin{theorem} \label{thm:f2m-43-cutedge}
There exists a G2M of cost at most $\frac 43$ times the cost of an F2M  if the F2M has no cut edge.
\end{theorem}

\begin{proof}
As described above, it is sufficient to focus on a single fractional component of the F2M.  Let $G$ be the support graph of this component.

To find the G2M, we find a minimum-cost perfect matching on the graph $G'$ we obtain by replacing each path in $G$ by a single edge, which we will call (at the risk of some confusion) a path edge. We set the cost of this edge to be the cost of the path in $G$, and we set the cost of a cycle edge in $G'$ to the {\it negative} of the cost of the cycle edge in $G$.  Note that $G'$ is cubic and 2-edge-connected because the support graph $G$ of the F2M has no cut edge.

Given a minimum-cost perfect matching in $G'$, we construct a G2M in $G$ by first including all paths from $G$. If a path edge is in the matching in $G'$, we double the path in $G$. If a cycle edge is {\it not} in the matching in $G'$, then we include the cycle edge in the G2M in $G$, otherwise we omit the cycle edge.

We first show that this indeed defines a G2M: for each vertex, the degree is four if the perfect matching contains the path edge incident on the vertex (since in that case, the two cycle edges on the vertex cannot be in the perfect matching, and hence both are added to the G2M together with two copies of the path), and it is two otherwise (since one cycle edge is in the perfect matching and hence only the other cycle edge and one copy of the path are added to the graphical 2-matching). Note that any connected component indeed has at least three nodes, since for any doubled path, we also take the four cycle edges incident on the endpoints.

We let $C$ denote the sum of the costs of the cycle edges, and $P$ the cost of the paths. Note that the cost of the F2M solution is $\frac 12 C + P$. The cost of the G2M is equal to the cost of all edges in the support graph ($P+C$) plus the cost of the perfect matching.  Because $G'$ is cubic and 2-edge-connected, we can invoke Lemma~\ref{lem:np} to show that the perfect matching has cost at most a third the cost of the edges in $G'$, or at most $\frac 13 P -\frac 13 C$.  Hence the cost of the G2M is at most $$P+C+ \frac 13 P - \frac 13 C = \frac 43 P + \frac 23 C = \frac 43 \left(P + \frac 12 C\right),$$ or at most $\frac 43$ the cost of the F2M solution, as claimed.
\end{proof}

\iftoggle{abs}{}{
The idea of using edges from a perfect matching to decide which edges to include in a matching and which edges to remove has also been used recently by M\"omke and Svensson~\cite{MomkeS11}.}

We now modify the proof of the theorem above so that the result extends to the case in which the F2M has cut edges.

\begin{theorem}[Boyd and Carr~\cite{BoydC99}]
\label{thm:bc43}
There exists a G2M of cost at most $\frac 43$ times the cost of an F2M.
\end{theorem}

\begin{proof}
As described above, it is sufficient to focus on a single fractional component of the F2M, and we let $G$ be the support graph of this component.

We once again create a new graph $G'$ from $G$, so that we can later define a matching problem in $G'$.  The matching will again show us how to create a G2M in $G$.  We extend the previous construction to deal with the case when the support graph has cut paths.
We introduce a gadget in $G'$ for each cut path in $G$, which replaces the cut path and its two endpoints. The other paths in $G$ are again replaced by single edges in $G'$ of cost equal to the cost of the path.  Each cycle edge in $G$ is also in $G'$ with cost equal to the negative of its cost in $G$.

To introduce the cut-path gadget, we begin by using an idea of Boyd and Carr~\cite{BoydC99}; namely, that we only need to consider three {\em patterns} to get an almost feasible graphical 2-matching on the cut path, when we allow ourselves to increase the cost by a third compared to the F2M.  Suppose the cut path has $\ell$ edges and $\ell+1$ nodes, and let $k=\lfloor \ell / 3 \rfloor$.
We can remove every third edge, double the remaining edges to obtain groups of nodes that are 2-edge-connected, where we get $k$ groups of three nodes that are G2M components, plus one group of $\ell-3k\in \{0,1,2\}$ nodes.
Alternatively, we could remove every third edge, starting from the first edge and double the remaining edges, in which case the first group has one node, the next $k$ or $k-1$ groups have three nodes and the last group again has one or two nodes.
The final pattern removes every third edge, starting from the second edge, so that the first group has two nodes, the next $k$ or $k-1$ groups have three nodes, and, again, the last group has one or two nodes.
Figure~\ref{fig:patterns} illustrates the three patterns for $\ell = 9$.

\iftoggle{abs}{
\begin{figure}[t]
\begin{minipage}[b]{0.5\linewidth}
\begin{center}
\includegraphics[width=.9\textwidth]{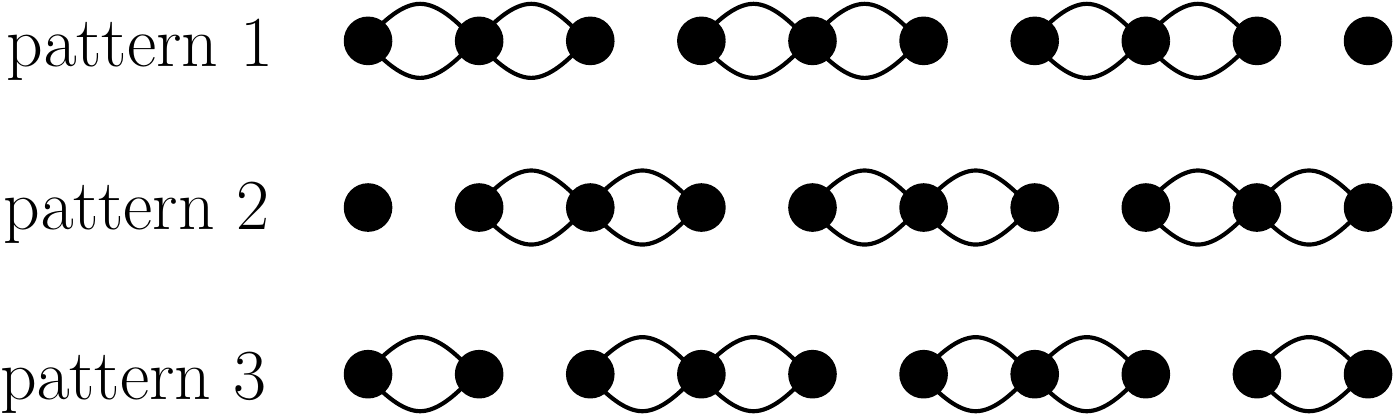}
\end{center}
\caption{Illustrations of patterns for $\ell=9$.}
\label{fig:patterns}
\end{minipage}
\hspace{0.5cm}
\begin{minipage}[b]{0.5\linewidth}
\begin{center}
\includegraphics[width=.9\textwidth]{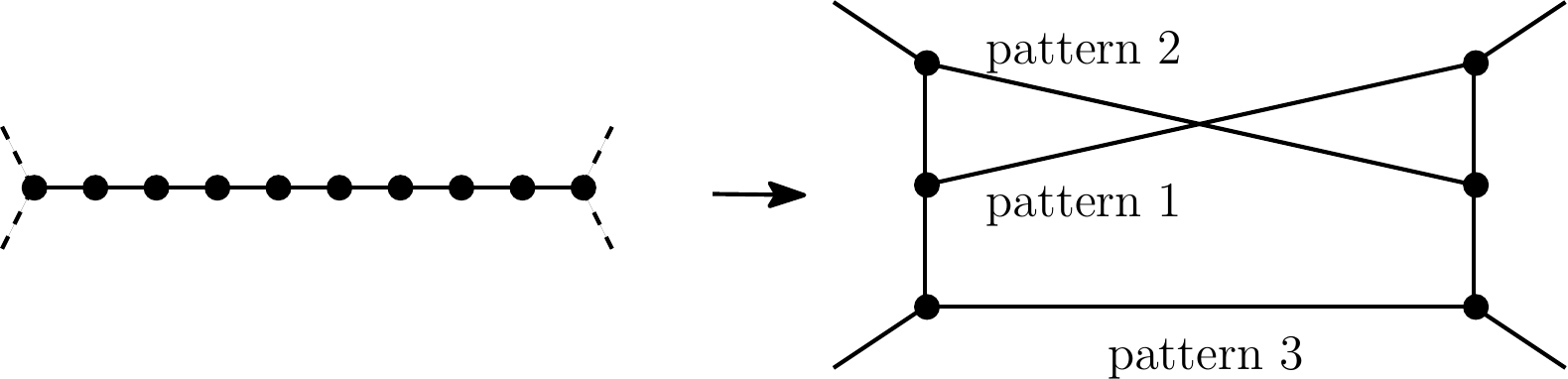}
\end{center}
\caption{Pattern gadget for $\ell=9$.}
\label{fig:pattern-gadget}
\end{minipage}
\end{figure}
}{
\begin{figure}
\begin{center}
\includegraphics[width=.5\textwidth]{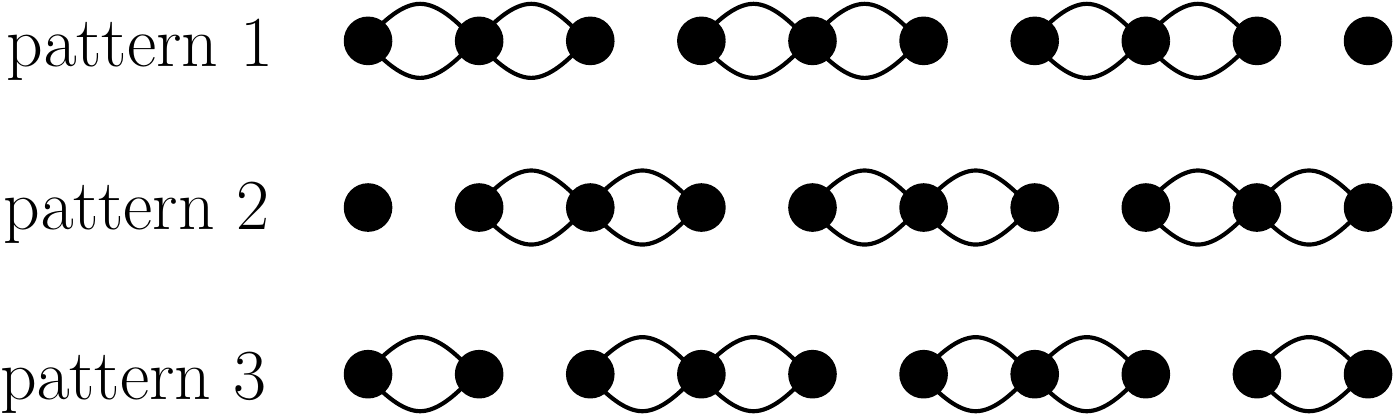}
\end{center}
\caption{Illustrations of patterns for $\ell=9$.}
\label{fig:patterns}
\end{figure}
}

To get a G2M that contains a certain pattern, we will ensure that if a group has size less than three, the G2M will include the two cycle edges incident on the first node (if the group is at the start of the pattern) or last node (if the group is at the end of the pattern).

We remark that there is exactly one pattern that starts with a group of size one, two and three, and hence two patterns need the G2M to include two cycle edges incident on the first node of the cut path. On the other hand,
there is also exactly one pattern that ends with a group of size one, two and three (the length of the cut path determines which of the three patterns ends with a group of size three: it is the second pattern if $\ell \pmod 3 = 0$, the third pattern if $\ell \pmod 3 = 1$ and the first pattern if $\ell \pmod 3 = 2$), and hence there are also two patterns that need the G2M to include the two cycle edges incident on the last node of the cut path.

We are now ready to define the cut-path gadget.
We replace each endpoint of the cut path in $G$ by a path of length two in $G'$; each of these new edges will have cost 0. Each node on the path will be connected to a {\it pattern edge} corresponding to one of the three patterns.
The middle node is connected to the pattern edge corresponding to the pattern which does not need two cycle edges incident on the endpoint of the cut path (i.e.\ the pattern for which the group containing the endpoint has size three).
We set the cost of a pattern edge to the cost of the edges in the corresponding pattern.  See Figure~\ref{fig:pattern-gadget} for an illustration of the gadget when $\ell = 9$.

\iftoggle{abs}{}{
\begin{figure}
\begin{center}
\includegraphics[width=.4\textwidth]{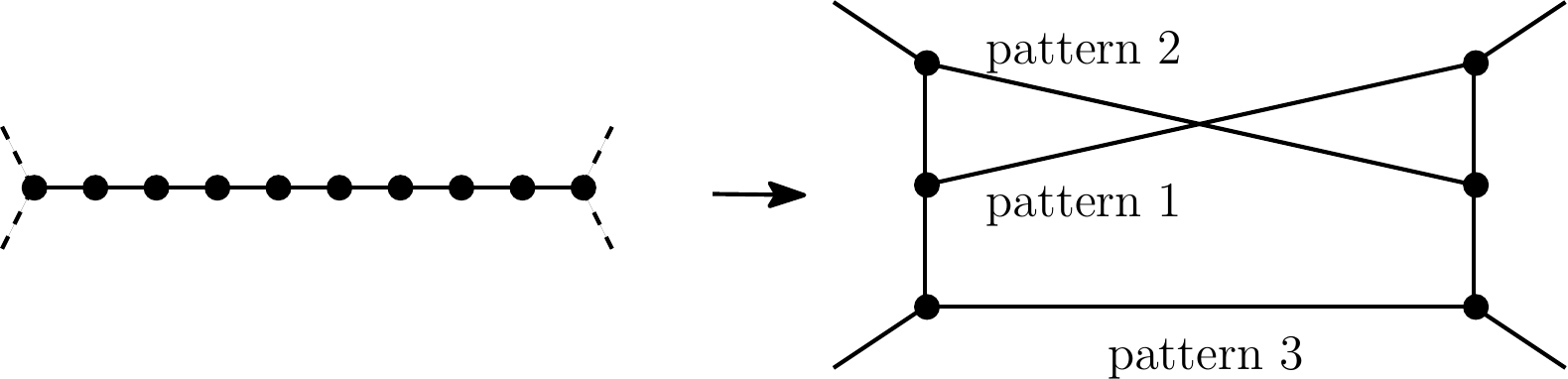}
\end{center}
\caption{Pattern gadget for $\ell=9$.}
\label{fig:pattern-gadget}
\end{figure}
}

If we replace each cut path in $G$ by a cut-path gadget in $G'$, once again $G'$ will be a cubic graph.  It is not hard to check that $G'$ is also 2-edge-connected because we have replaced the cut path in $G$ with three pattern edges crossing the cut in $G'$.

We argue that there is a minimum-cost perfect matching that uses exactly one edge from each cut-path gadget.
Note that the fact that we replace only the cut paths in $G$ by a cut gadget in $G'$ means that a perfect matching in $G'$ contains an odd number of pattern edges in a gadget. If it contains three pattern edges, then we could find a matching of no greater cost by choosing only one pattern edge, namely the pattern edge that is not incident on the middle node for the either one of its endpoints. Note that we can add two edges of cost 0 that connect the four nodes incident on the other two pattern edges, to again have a perfect matching without increasing the cost.

Now we show how to obtain a G2M in $G$ from the minimum-cost perfect matching in $G'$.  In the G2M we include
all edges from $G$ that are in paths which are not cut paths, the cycle edges in $G$ which are {\it not} chosen by the perfect matching,
duplicates of edges in paths in $G$ that are chosen by the perfect matching,
and the edges in a pattern if the corresponding pattern edge is in the perfect matching.

We argue that this set of edges is a G2M in $G$.  Note that if the perfect matching contains only the pattern edge incident on the middle node, then the two cycle edges that are adjacent to the gadget are also in the matching. Hence the corresponding endpoint in $G$ of the cut path has no cycle edges incident on it in the G2M, but since the pattern edge is incident on the middle node, the corresponding pattern ensures that the node has degree two and is in a connected component of size three.
If the perfect matching contains the pattern edge incident on a node other than the middle node, then neither of the two cycle edges that are adjacent to the gadget in $G'$ are in the perfect matching. Hence the corresponding endpoint of the cut path in $G$ has both of these cycle edges incident on it in the G2M, and zero or two edges from the pattern corresponding to the chosen pattern edge.  Hence the node has degree two or four and it is in a connected component of size at least three.

As before, because $G'$ is cubic and 2-edge-connected, we can apply Lemma~\ref{lem:np} to bound the cost of the perfect matching in $G'$.
Let $P_1$ be the cost of the paths in $G$ that are not cut paths, and $P_2$ the cost of the cut paths in $G$, so that the cost of the F2M is $P_1 + P_2 + \frac 12 C$. Note that the cost of the three pattern edges in the gadget corresponding to a cut path sums up to four times the cost of the cut path.  Thus the total cost of the edges in $G'$ is $P_1 + 4P_2 - C$.  By Lemma~\ref{lem:np}, the cost of the perfect matching in $G'$ is at most $\frac 13 P_1 + \frac 43 P_2 - \frac 13 C$.
The cost of the G2M corresponding to the minimum-cost perfect matching is therefore at most $$P_1 + \frac13 P_1 + \frac 43 P_2 + C-\frac13 C = \frac 43 P + \frac 23 C = \frac 43 \left(P + \frac 12 C \right)$$ as claimed.
\end{proof}

\iftoggle{abs}
{We can extend the ideas above to obtain a better G2M if no cut paths exist.  The basic idea is that we replace every path by a cut-path gadget, and show that the solution $x(e)=\frac 1 9$ if $e$ is a pattern edge and $x(e) = \frac 4 9$ if $e$ is a cycle edge is feasible for the matching polytope $(M)$.}
{We now show how to use the ideas behind the cut-path gadget to obtain a better G2M if no cut paths exist.}

\begin{theorem}
\label{thm:f2m-109}
If an F2M has no cut edge, then there exists a G2M of cost at most $\frac {10}{9}$ times the cost of the F2M.
\end{theorem}

\iftoggle{abs}{}{
\begin{proof}
Once again we define a new graph $G'$ from the support graph $G$ of a fractional component of the optimal F2M.  Each cycle edge in $G$ is in $G'$ with cost that is the negative of its cost in $G$.  Each path in $G$ and its two endpoints are replaced by the cut-path gadget used in the proof of Theorem~\ref{thm:bc43}.  The costs of the pattern edges in $G'$ are slightly different than in the previous proof: we subtract the cost of the original path from the cost of each pattern edge in its gadget.
In other words, the cost of a pattern edge in $G'$ is obtained by adding once the cost of the edges that appear twice in the pattern and subtracting the cost of the edges that do not appear in the pattern.
Note that the sum of the costs of the three pattern edges in $G'$ is equal to the cost of the original path in $G$.
Also, note that the sum of the costs of any two pattern edges in $G'$ is nonnegative: an edge on the path contributes its cost either positively to one pattern and negatively to the other, or positively to both patterns.

We first argue that there is a minimum-cost perfect matching that chooses either zero or one pattern edge in each cut-path gadget. Suppose the perfect matching contains two pattern edges in a gadget.  Note that on both sides of the gadget these pattern edges must be incident on the middle node, otherwise some middle node is not matched. Hence the four endpoints of the two pattern edges are connected in $G'$ by two edges of cost zero. By the observation above, the cost of the two pattern edges is nonnegative, and so we can remove the two pattern edges from the matching and add the two edges of cost zero without increasing the cost of the matching.
By the same argument, we can handle the case that the perfect matching contains three pattern edges from a gadget by choosing the pattern edge that is not incident on the middle node on both sides of the gadget, and replacing the other two pattern edges in the matching by the cost zero edges that connect their endpoints.

Therefore, we can assume the perfect matching chooses either zero or one pattern edge in a gadget. If it chooses zero pattern edges, then we add the path from $G$ to the G2M. Otherwise, the pattern corresponding to the chosen pattern edge is added to the G2M.
We also add the cycle edges to the G2M corresponding to the cycle edges that are {\it not} in the perfect matching.

By almost the same arguments as before, the solution constructed is indeed a G2M.  The only case not covered by previous arguments is the case in which zero pattern edges are chosen in $G'$.  Then it must be the case that one of the two cycle edges is chosen in $G'$ and the other is not, so that one of the two cycle edges is included in the G2M and the other not.  Since we include the path from $G$ in the G2M if no pattern edges are chosen, the endpoint of the path will have degree two.

To argue about the cost of the minimum-cost perfect matching in $G'$, we create a feasible solution for the matching linear program $(M)$.  To do this, for each pattern edge $e$, we set $x(e) = \frac 19$, and for every other edge $e'$, we set $x(e')=\frac 49$.  We will show this is a feasible solution in a moment.  Let $P$ be the cost of the path edges in the F2M, and $C$ the cost of the cycle edges, so that the F2M has cost $P + \frac 12 C$.  Since the sum of the cost of the pattern edges in a gadget is equal to the cost of the path, the cost of this solution for $(M)$ is $\frac 19 P -\frac 49 C$, and there exists a perfect matching of cost at most this much.  Thus the cost of the G2M is at most $$P + \frac 19 P + C -\frac 49 C = \frac {10}9P+\frac 59 C = \frac {10}{9} \left(P + \frac 12 C \right),$$ as claimed.

To see that $x$ is a feasible solution for $(M)$, consider any cut such that the number of nodes on each side of the cut is odd. If there exists a cycle from the F2M such that not all nodes in the gadgets for the nodes in the cycle are on the same side of the cut, then there are two edges crossing the cut with value $\frac 49$.  Since $G'$ is cubic, if the cut has odd size, then the total number of edges crossing the cut is odd, and there must be at least one more edge in the cut with value at least $\frac 19$.  Hence the total value on the edges crossing the cut is at least one.  For any other cut, since there is no cut path in $G$, there are at least three gadgets crossing the cut in $G'$. Since each gadget contains three pattern edges, the value of the edges crossing the cut is again at least one.
\end{proof}
}

\section{A Polyhedral Proof of the Boyd-Carr Conjecture}
\label{sec:2m-poly}

We will generalize the result in Theorem~\ref{thm:f2m-109} and show that the ratio between the cost of the optimal 2-matching and the subtour LP is at most $\frac{10}9$. In the combinatorial proofs of the previous section, we heavily used the fact that F2Ms have a nice simple structure, and, unfortunately, this does not hold for the subtour LP solution. We therefore turn to a polyhedral rather than a combinatorial proof. We derive a polyhedral description for graphical 2-matchings, and we then use this description to construct a feasible (fractional) G2M solution from any solution to the subtour LP of cost not more than $\frac{10}{9}$ times the value of the subtour LP. The manner in which the feasible G2M solution is defined based on a solution to $(\SUBT)$ is a generalization of the proof of Theorem~\ref{thm:f2m-109}.

\newcommand{\man}{\ensuremath{\mathrm{man}}}
\newcommand{\opt}{\ensuremath{\mathrm{opt}}}
We start by giving a polyhedral description of a generalization of 2-matching, where the node set consists of ``mandatory nodes'' ($V_\man$) and ``optional nodes'' ($V_\opt$). The former need to have degree $2$ in the solution, whereas the latter can have degree $0$ or $2$. We will refer to this problem as the 2-Matching with Optional Nodes Problem (2MO).
\begin{theorem} \label{2fo}
Let $G = ( V_\man \cup V_\opt, E )$ be a 2MO instance. The convex hull of integer 2MO solutions is given by the following polytope:
\lps
& & & & \sum_{ e \in \delta( i ) } y(e) = 2, & \forall i \in V_\man, \numb{2fodegreeconsman}\\
& & & & \sum_{ e \in \delta( i ) } y(e) \leq 2, & \forall i \in V_\opt, \numb{2fodegreeconsopt} \\
& & & & \sum_{ e \in \delta( S ) \setminus F } y(e) + \sum_{ e\in F } ( 1 - y(e) ) \geq 1, &\forall S \subseteq V,\, F \subseteq \delta( S ),\, F
\text{ matching},\, |F| \text{ odd,} \numb{2foFcons} \\
& & & & 0 \leq y(e) \leq 1, &\forall e\in E. \numb{2foboundscons} \elps
\end{theorem}
\iftoggle{abs}{The proof of Theorem~\ref{2fo} is similar to the proof of the polyhedral description of the 2-matching polytope (Theorem 30.8) in Schrijver~\cite{Schrijver-book}.}{The proof of Theorem~\ref{2fo} is similar to the proof of the polyhedral description of the 2-matching polytope (Theorem 30.8) in Schrijver~\cite{Schrijver-book}, and is deferred to Appendix~\ref{sec:2foproof}.}

Recall the definition of a graphical 2-matching (G2M): (i) each vertex has degree either 2 or 4, (ii) each edge has 0, 1, or 2 copies, and (iii) each component has size at least three. We will (for the moment) relax the second condition so that each edge has at most 3 copies.

\begin{lemma}\label{lem:red}
We can reduce a G2M instance $G = ( V, E )$ to a 2MO instance $G' = ( V', E' )$ as follows:  Let $V'_\man = \{i_m: i \in V \}, V'_\opt = \{ i_o: i \in V \}, V' = V'_\man \cup V'_\opt, E' = \{ (i_m,j_m) : (i, j) \in E \} \cup \{ (i_m,j_o) : (i,j) \in E \}$. We add an edge $\{ i, j \}$ to the (relaxed) G2M solution for each edge $(i_m,j_m)$, $(i_o,j_m)$ and $(i_m,j_o)$ that is in the associated 2MO solution.
\end{lemma}

\iftoggle{abs}{}{
\begin{proof}
Note that condition (i) for node $i$ directly follows from the degree constraints for nodes $i_m$ and $i_o$ in the reduction. Relaxed condition (ii) follows from the fact that for every edge in the G2M instance there are three associated edges in the 2MO instance. Finally, since each node $i_m$ has degree $2$ in the 2MO solution, there cannot be a component of size $1$. Suppose there there is a component of size 2. Then this must be an isolated doubled or quadrupled edge, say $(i,j)$, because of the degree constraints.  Clearly we can't have a quadrupled edge since there are at most three copies of edge $(i,j)$ in the 2MO solution.  We also can't have an isolated doubled edge: in order for the edge to be isolated, we would need $(i_m, j_m)$ and $(i_m, j_o)$ to be in the 2MO solution.  But then $j_o$ must have degree 2, and its second edge must be $(j_o, k_m)$ for some $k\neq i,j$, since there are no edges $(i_o,j_o)$ or $(j_o,j_m)$ in the 2MO instance.
\end{proof}
}

If the edges have nonnegative costs, we may assume with loss of generality that each edge appears at most twice in an optimal G2M solution: if any edge appears three times, we can remove two copies of it without affecting the parity of its endpoints, and the cost cannot increase.

\iftoggle{abs}{
The following lemma shows how to map a solution $x$ of the subtour LP to a solution $y$ to the 2MO polytope corresponding to a G2M.  The mapping is based on some insights gleaned from the proof of Theorem \ref{thm:f2m-109}; we omit this discussion due to space constraints.
}{
We will now use a solution to the subtour LP on $G=(V,E)$ to define a feasible solution to the 2MO instance $G'=(V',E')$ associated with the graphical 2-matching problem on $G$. It will be instructive to first consider the case when the subtour LP solution $x$ is an F2M with no cut edge. In that case, the proof of Theorem~\ref{thm:f2m-109} gives us a way to construct a G2M solution. In fact, it allows us to find a probability distribution on G2Ms, such that the expected cost of the G2M is exactly $\frac{10}9$ times the cost of the F2M solution.
This probability distribution has a number of special properties: (i) if a G2M has positive probability, then each doubled edge is a path edge with $x$-value 1, and has exactly one endpoint that has degree 4; (ii) for each path edge $(i,j)$ with $x$-value 1, the probability that it occurs twice and $i$ has degree 4 is $\frac 1 9$, and the expected number of times $(i,j)$ occurs is $\frac{10}9$.
These observations give a hint as to how we should define a 2MO solution based on a subtour LP solution $x$. We think of the edge $(i_m, j_m)$ as the first copy of the edge $(i,j)$, and $(i_m,j_o)$ as the second copy if $j$ has degree 4, and $(i_o,j_m)$ as the second copy if $i$ has degree 4. Then the probability of $(i_m,j_o)$ and $(i_o,j_m)$ is $\frac 19 x(i,j)$ if $x(i,j)=1$, and the probability of $(i_m,j_m)$ is $\frac 89 x(i,j)$.
This interpretation does not quite work for the cycle edges $(i,j)$ with $x$-value $\frac 12$, since at most one copy occurs in the G2M.

A better interpretation is that we consider a Eulerian walk on each component of the G2M solution, and associate $i_m$ with the first time we enter and leave node $i$, and $i_o$ with the second time we enter and leave node $i$ (if $i$ has degree 4). If we direct the walk in each of the two possible directions with probability $\frac 12$, then the probability we use edge $(i_m,j_m)$ is $\frac 89 x(i,j)$ and the probability we use edge $(i_m,j_o)$ is $\frac 19 x(i,j)$.  We argue this as follows.

For a path edge $(i,j)$ with $x(i,j)=1$, the probability that we use edge $(i_m,j_o)$ is $\frac 19$, since if $j$ has degree 4, we know by the construction that $(i,j)$ is a doubled edge, and $i$ has degree 2. Hence, if $j$ has degree 4, then $(i_m,j_o)$ is in the walk, and the probability that $j$ has degree 4 is $\frac 19$. A similar argument shows that we use $(i_o,j_m)$ with probability $\frac 19$.
Also, the expected number of times we use edge $(i,j)$ in the G2M is $\frac {10}9$, so the probability of using $(i_m,j_m)$ in the walk must be $\frac 89$.

For a cycle edge $(i,j)$ with $x(i,j)=\frac 12$, the probability that we use $(i_m,j_o)$ is $\frac 19 \cdot\frac 12$, since if $j$ has degree 4, then the G2M contains a doubled path edge $(j,k)$ where $x(j,k)=1$ and $k$ has degree 2. Hence the probability that we use $(i_m,j_o)$ is the probability that the walk is directed in such a way that we visit $i$ before the loop from $j$ to $k$ and back, and this happens with probability $\frac 12$. Similarly, the probability that we use $(i_o,j_m)$ is $\frac 1{18}$, and the fact that the expected number of times we use an edge with $x$-value $\frac 12$ in the G2M is $\frac 59$, shows that the probability of using $(i_m,j_m)$ in the walk must be $\frac 49$.

The following lemma states that using the probabilities $\frac 89 x(i,j)$ and $\frac 19 x(i,j)$ to define a fractional solution to the 2MO instance corresponding to the G2M instance $G$ also yields a feasible solution if, rather than an F2M with no cut edge,  $x$ is a feasible solution to the subtour LP on $G$.
}

\begin{lemma} \label{lem:boydcarr}
Given a graph $G = ( V, E )$,
let $x$ be a feasible solution to the subtour LP for $G$. Then the following solution is a feasible solution to the 2MO instance $G'=(V',E')$ associated with the graphical 2-matching instance given by $G$ for $\alpha=\frac 1 9$:
\iftoggle{abs}{ $y( i_m, j_m ) = (1-\alpha) x( i,j )$, $y( i_m, j_o ) = \alpha x( i,j )$, $y ( i_o, j_m ) = \alpha x( i,j )$ }
{
\begin{align*}
    y( i_m, j_m ) &= (1-\alpha) x( i,j ) \\
    y( i_m, j_o ) &= \alpha x( i,j ) \\
    y ( i_o, j_m ) &= \alpha x( i,j )
\end{align*}
}
for all $(i,j) \in E$.
\end{lemma}

Note that the cost of the constructed G2M solution is exactly $\frac {10}9$ times the cost of the solution of the subtour LP.  Thus our result follows immediately from the lemma.

\begin{corollary} \label{cor:boydcarr}
There exists a G2M of cost at most $\frac {10}9$ times the value of the subtour LP.
\end{corollary}

\begin{proofof}{Lemma~\ref{lem:boydcarr}}
We need to show that $y$ satisfies the constraints~(\ref{2fodegreeconsman})-(\ref{2foboundscons}) on $G'$, where $G'$ is defined as in Lemma~\ref{lem:red}.  Constraints~(\ref{2fodegreeconsman}),  (\ref{2fodegreeconsopt}) and (\ref{2foboundscons}) are obviously met, and we only need to show that constraints~(\ref{2foFcons}) are met.  To this end, fix $S \subseteq V'$, $F \subseteq \delta( S )$ where $F$ is a matching and $|F|$ is odd. We define $z( e' ) = y( e' )$ if $e'\in \delta(S)\backslash F$ and $z(e')=1-y(e')$ if $e'\in F$.  For simplicity, for any set of edges $X \subseteq E'$, we define $z(X) = \sum_{e' \in X} z(e')$.
Then we need to show that $z(\delta(S)) \ge 1$.

First, suppose $S$ does not contain any node $i_m$ for any $i\in V$.
For any $j_o \in S$, we have that $z(\delta(S)\cap \delta(j_o)) = z(\{(i_m,j_o):i\in V\})$.
Since $|F|\ge 1$, there exists some $j_o\in S$ such that $F$ contains some edge incident on $j_o$, say $(i'_m, j_o)$.
Then, $z(\{(i_m,j_o):i\in V\}) = 1-\alpha x(i',j) + \sum_{i\in V: i\neq i'} \alpha x(i,j) = \alpha x(\delta(j)) + 1 - 2\alpha x(i',j)$.
Now, note that $x(\delta(j))=2$ and $x(i',j)\le 1$, hence $z(\delta(S)\cap \delta(j_o))\ge 1$.

By symmetry, it remains to consider the case when both $S$ and $V'\backslash S$ contain a node $i_m$ for some $i\in V$.

We consider an edge $e=(i,j)\in G$ such that at least one of the three edges $(i_o,j_m), (j_m, i_m), (i_m,j_o)$ crosses the cut $S$ in $G'$.
Note that there are $2^3-1=7$ possible choices for the edges that cross the cut. We discern five different types of edges in $G$ for which at least one of the three corresponding edges crosses the cut (type II and type V each cover 2 of the possible choices):
\begin{itemize}[noitemsep]
\item[(I)] The edge $(i_m, j_m)$ crosses the cut.
\item[(II)] The edges $(i_o,j_m)$ and $(j_m, i_m)$ or the edges $(j_m, i_m)$ and $(i_m, j_o)$ cross the cut.
\item[(III)] The edges $(i_o,j_m), (j_m, i_m)$ and $(i_m, j_o)$ cross the cut.
\item[(IV)] The edges $(i_o,j_m), (i_m, j_o)$ cross the cut.
\item[(V)] The edge $(i_o,j_m)$ or the edge $(i_m, j_o)$ crosses the cut.
\end{itemize}
\iftoggle{abs}
{Figure~\ref{fig:types} in Appendix \ref{sec:figs} illustrates the five types.}
{Figure~\ref{fig:types} illustrates the five types.
\begin{figure}
\begin{center}
\subfloat[Type I.]{\includegraphics[width=.22\textwidth]{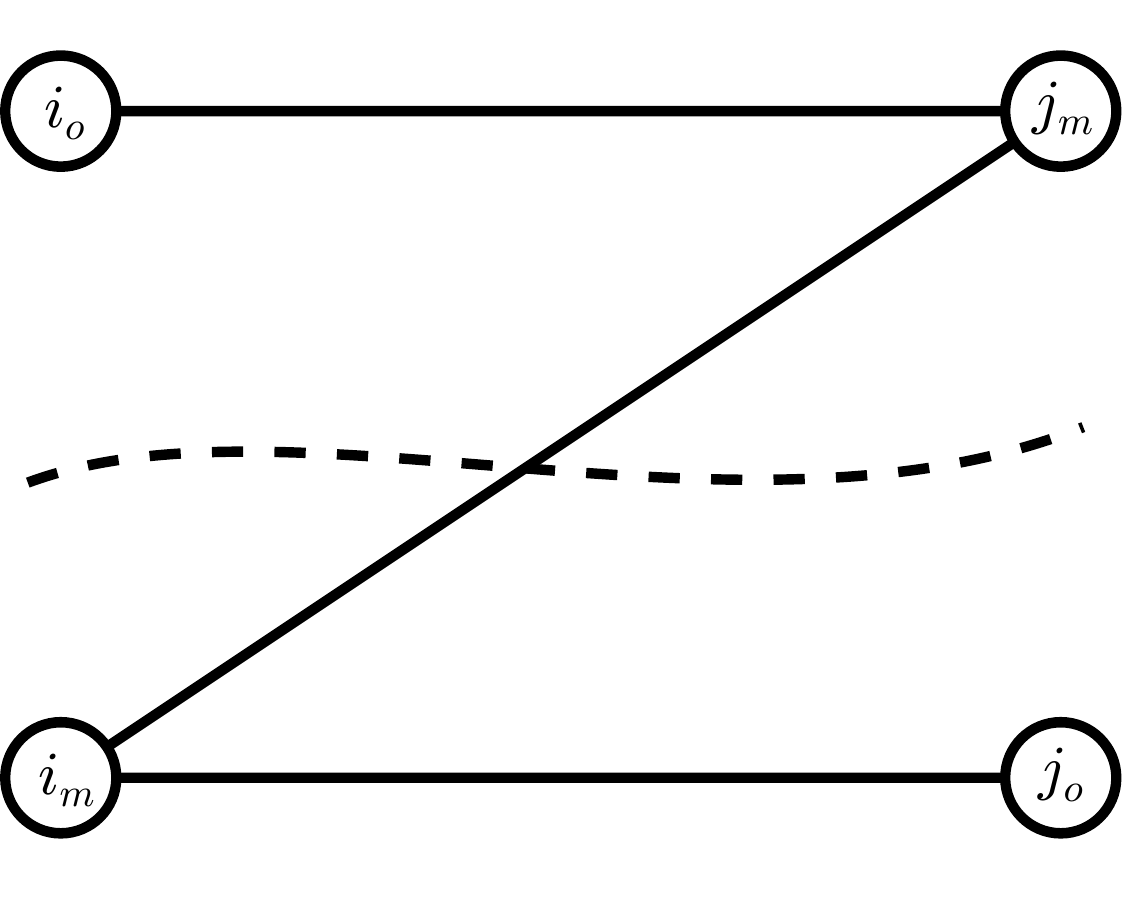}}\qquad
\subfloat[Type II.]{\includegraphics[width=.22\textwidth]{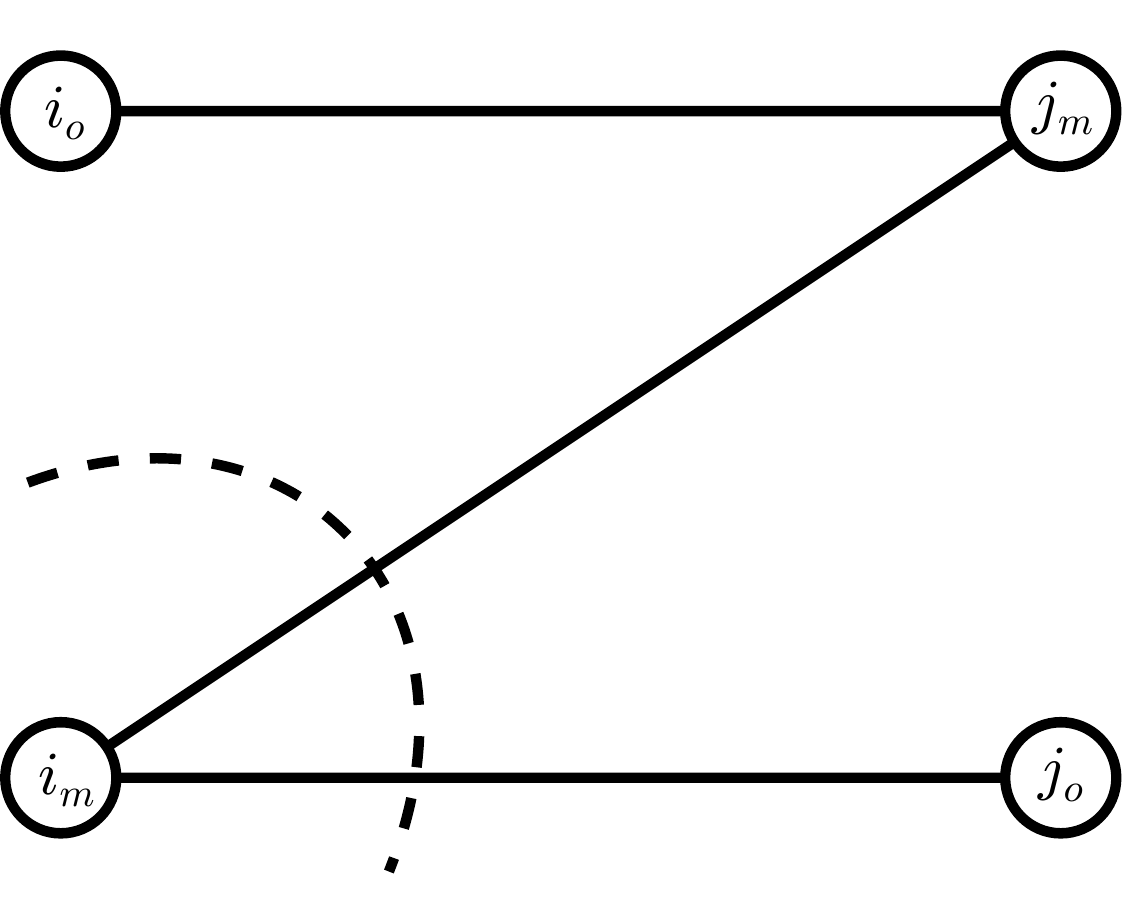}}\qquad
\subfloat[Type III.]{\includegraphics[width=.22\textwidth]{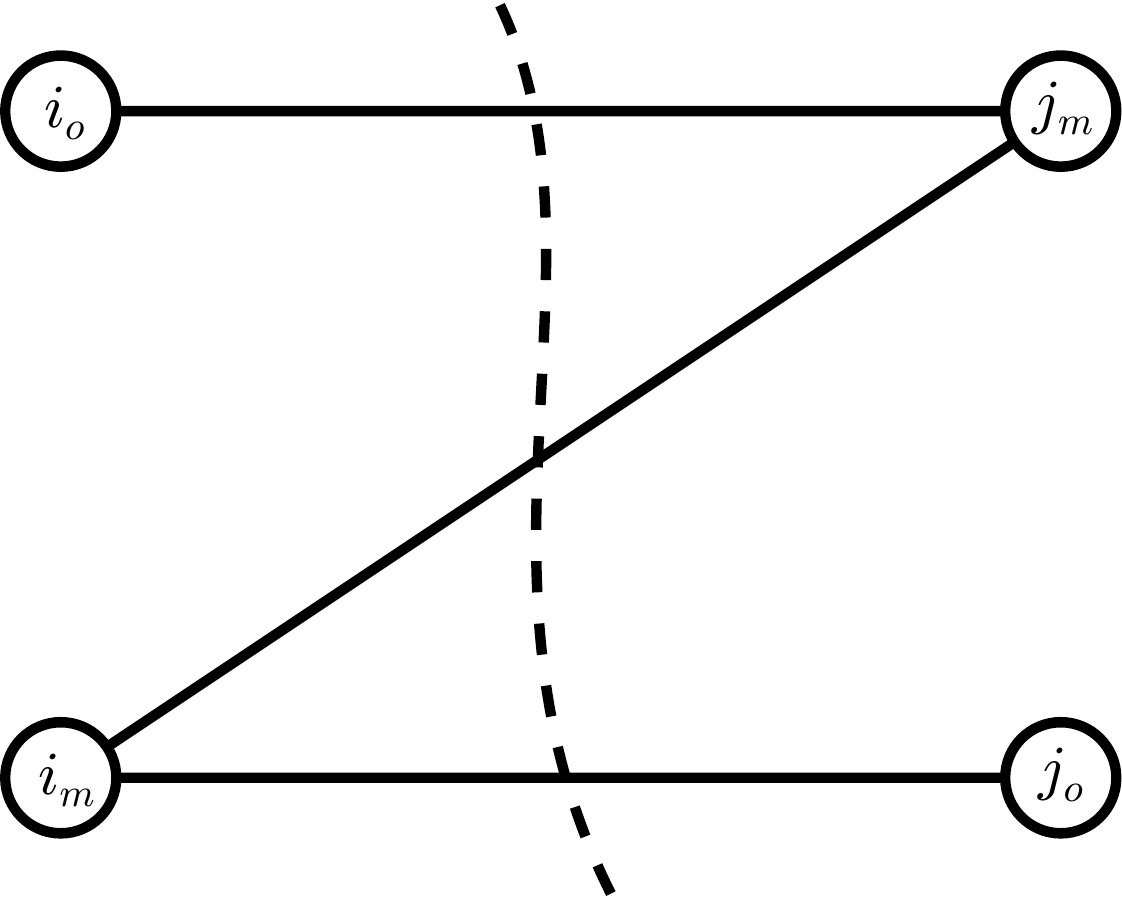}}\\
\subfloat[Type IV.]{\includegraphics[width=.22\textwidth]{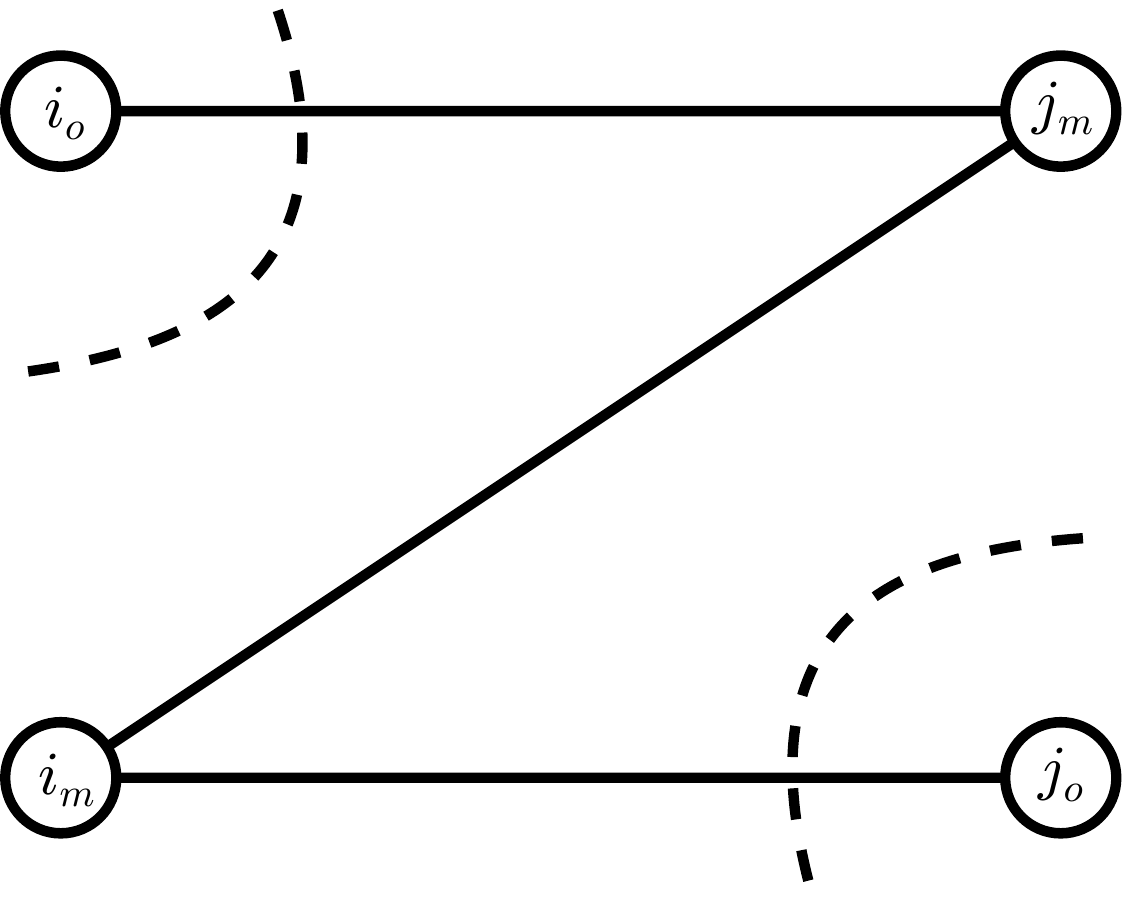}}\qquad
\subfloat[Type V.]{\includegraphics[width=.22\textwidth]{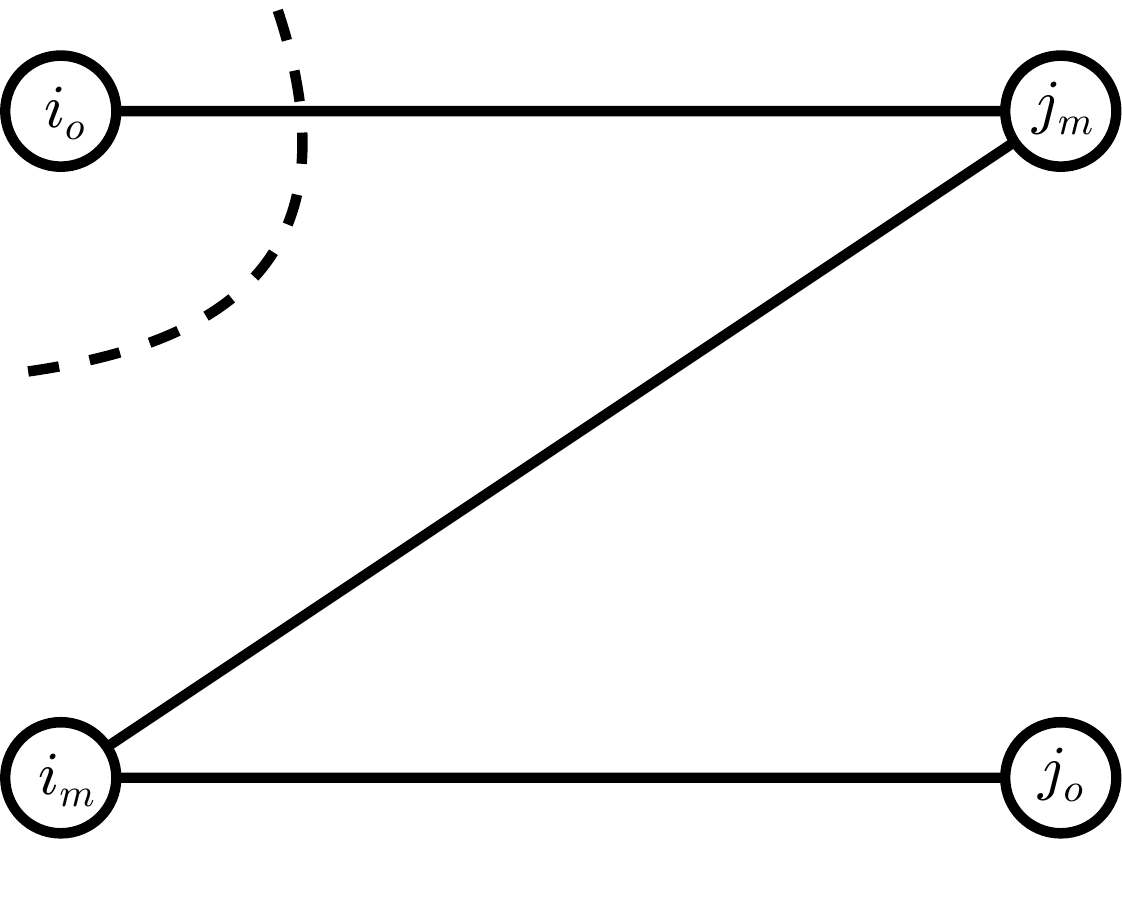}}\qquad
\end{center}
\caption{Illustrations of the five types of cuts of the edges in the reduction. The $y$-value on the top and bottom edge is $\alpha x(i,j)$ and the $y$-value on the middle edge is $(1-\alpha)x(i,j)$.}
\label{fig:types}
\end{figure}
}

We use the notation $i_*$ to denote either $i_m$ or $i_o$, and we will say an edge $e'=(i_*,j_*)\in G'$ is in a gadget of type I, II,~\ldots, V, if the edge $(i,j)\in G$ is an edge of that type.

We now consider three different cases, depending on the set $F$.

\begin{claim}\label{claim1}
If $F$ contains an edge in a gadget of type IV or V, then $z(\delta(S)) \ge 1$.
\end{claim}

\iftoggle{abs}{}{
\begin{proof}
Let $e'\in F$ be contained in a gadget of type IV or V. Note that $e'$ has one endpoint in $V'_\man$ and one endpoint in $V'_\opt$.
Let $e'=(i_o,j_m)$. Since $(j_m, i_m)$ does not cross the cut, $i_o$ and $i_m$ are on different sides of the cut.

Hence, the paths $\{(i_o,j'_m), (j'_m, i_m)\}$ cross the cut $S$ for every $j'\in V$.
Each of these paths thus contribute at least $\alpha x(i,j')$ to $z(\delta(S))$ for $j' \neq j$.
Also, since $e'=(i_o,j_m)\in F$, $z(e')= 1- \alpha x(i,j)$.
We thus get that $z(\delta(S)) \ge \sum_{j'\neq j} \alpha x(i,j') + 1-\alpha x(i,j) = \sum_{j'} \alpha x(i,j') + 1 - 2\alpha x(i,j)\ge 1$, where the last inequality follows since $\sum_{j'}x(i,j')=2$ by the degree constraints, and $x(i,j)\le 1$.
\end{proof}
}

For the remaining cases, we associate a cut $R$ in the graph $G$ with the cut $S$ in $G'$:
let $R=\{i\in V: i_m\in S\}$. Note that $R, V\backslash R$ are not empty.
Note that if $e$ is of type I, II, or III, then the edge $(i_m, j_m)$ crosses the cut, and hence, the edge $e$ crosses the cut $R$ in~$G$.

In the remainder of this proof, we will write $z(\delta(S))= y(\delta(S)) + |F| - 2 y(F)$, and we will give a lower bound on $y(\delta(S))$ to show that $z(\delta(S))\ge 1$.
In order to give a lower bound on $y(\delta(S))$, we need to use the fact that $x$ satisfies degree constraints for each node, and that $x(\delta(R))\ge 2$.
It will therefore be convenient to relate the contribution to $y(\delta(S) )$ of the three edges $(i_o,j_m)$, $(j_m,i_m)$, and $(i_m, j_o)$ to the edge $(i,j)\in G$, if $(i,j)\in \delta(R)$, but also to the nodes $i$ and $j$ for certain types of nodes $i,j\in V$.

In particular, we say a node $i\in V$ is a {\em lonely node} if $|\{i_m,i_o\}\cap S|=1$. We let $L$ be the set of lonely nodes.
We assign each lonely node $i$ an amount of $\alpha x(i,j)$, for each edge $(i,j)$ of type I, II,~\ldots, V.
Note that for each lonely node $i$, the paths $\{(i_o, j_m), (j_m, i_m)\}$ cross the cut for all $j\in V$, and hence, each lonely node gets assigned
$\alpha \sum_j x(i,j)$, which by the degree constraints is equal to $2\alpha$.

\begin{itemize}[noitemsep]
\item[(I)]
For an edge $(i,j)$ of type I, the total contribution of the three edges $(i_o,j_m), (j_m,i_m), (i_m, j_o)$ to $y(\delta(S))$ is $(1-\alpha)x(i,j)$.
Note that both $i$ and $j$ are lonely nodes.
We assign $(1-3\alpha)x(i,j)$ to the edge $(i,j)$, and $\alpha x(i,j)$ each to nodes $i$ and $j$.
\item[(II)]
For an edge $(i,j)$ of type II, the total contribution of the three edges $(i_o,j_m), (j_m,i_m), (i_m, j_o)$ to $y(\delta(S))$ is $x(i,j)$.
Note that only one of $i,j$ is a lonely node, and we therefore assign $(1-\alpha)x(i,j)$ to the edge $(i,j)$, and $\alpha x(i,j)$ to the lonely node among $i,j$.
\item[(III)]
For an  edge $(i,j)$ of type III, the total contribution of the three edges $(i_o,j_m), (j_m,i_m), (i_m, j_o)$ to $y(\delta(S))$ is $(1+\alpha)x(i,j)$, and neither $i$ nor $j$ is a lonely node. We therefore assign $(1+\alpha)x(i,j)$ to the edge $(i,j)$.
\item[(IV)]
For an  edge $(i,j)$ of type IV, the total contribution of the three edges $(i_o,j_m), (j_m,i_m), (i_m, j_o)$ to $y(\delta(S))$ is $2\alpha x(i,j)$. Since $(i,j)\not \in \delta(R)$ and both $i$ and $j$ are lonely nodes, we assign 0 to $(i,j)$ and $\alpha x(i,j)$ each to $i$ and $j$.
\item[(V)]
For an  edge $(i,j)$ of type V, the total contribution of the three edges $(i_o,j_m), (j_m,i_m), (i_m, j_o)$ to $y(\delta(S))$ is $\alpha x(i,j)$. Since $(i,j)\not \in \delta(R)$ and only one of $i$ and $j$ is a lonely node, we can assign 0 to $(i,j)$ and $\alpha x(i,j)$ to the lonely node.
\end{itemize}

By the argument above, we have assigned $2\alpha$ to each lonely node. We now show how this fact, combined with the fact that $x(\delta(R))\ge 2$ and the assignment of values to the edges in $\delta(R)$, allows us to conclude that $z(\delta(S))\ge 1$.

\begin{claim}
If $|F|=1$, then $z(\delta(S))\ge 1$.\end{claim}

\iftoggle{abs}{}{
\begin{proof}
Let $F=\{ e' \}$.
Let $(i,j)$ be such that $e'=(i_*, j_*)$. We will show that $z(\delta(S)) = y(\delta(S))+1-2y(e') \ge 1$.
Note that $2y(e') \le 2(1-\alpha) x(i,j) \le 2-2\alpha$, so it is enough to show that $y(\delta(S)) \ge 2-2\alpha$.

First, suppose that $|L|\le 1$. Then, there is no edge of type I, so to each edge $e\in \delta(R)$, we assigned at least $(1-\alpha) x(e)$.
Hence, $y(\delta(S)) \ge (1-\alpha) x(\delta(R))\ge 2-2\alpha$, since $x(\delta(R))\ge 2$ by the subtour elimination constraints.

If $|L|\ge 2$, then we assigned $2\alpha$ to each node in $L$, giving at least $4\alpha$. We assigned at least $(1-3\alpha)x(e)$ to each edge $e\in \delta(R)$. Therefore, $y(\delta(S)) \ge 4\alpha + (1-3\alpha) x(\delta(R))\ge 2-2\alpha$, where we again use that $x(\delta(R))\ge 2$.
\end{proof}
}

\begin{claim}
If $|F|\ge 3$, then $z(\delta(S))\ge 1$.\end{claim}
\begin{proof}
By Claim~\ref{claim1}, we may assume that all edges in $F$ are contained in a gadget of type I, II or III, and hence, that the corresponding edges in $e\in G$ are in $\delta(R)$.
Let $E_1, E_2, E_3$ be the edges in $\delta(R)$ of type I, II and III, respectively, for which the gadget contains one or more edges in $F$.

Note that a lonely node $i$ can be incident on at most one edge in $E_1\cup  E_2\cup E_3$: Only the edges $(i,j)\in E_1\cup E_2$ can be incident on a lonely node $i$, and in the first case, $(i_m,j_m)$ must be in $F$, and in the second case, either $(i_m,j_o)$ or $(i_m,j_m)$ is in $F$, since these are the only edges that cross the cut for these types. Now, since $F$ is a matching, it can have at most one edge incident on $i_m$ and hence $i$ can be incident on at most one edge in $E_1\cup  E_2\cup E_3$.

We therefore have that
\[y(\delta(S)) \ge (1-3\alpha) x(E_1) + 4\alpha |E_1| + (1-\alpha) x(E_2) + 2\alpha |E_2| + (1+\alpha) x(E_3).\]

On the other hand, since $F$ is a matching, only the gadgets for edges of type III can contain two edges in $F$. Hence, $|F|=|E_1|+|E_2|+(1+\beta)|E_3|$, where $\beta$ is the fraction of edges in $E_3$ for which two edges in the corresponding gadget are contained in $F$.

Also, $y(F) \le (1-\alpha)\left(x(E_1)+x(E_2)+x(E_3)\right)$, since $y((i_*,j_*))\le (1-\alpha) x(i,j)$, and, if two edges in the gadget for $e\in E_3$ are contained in $F$, then these edges both have $y$-value $\alpha x(e)$, and since $\alpha \le \frac 13$, $2\alpha x(e) \le (1-\alpha) x(e)$.

Hence, we get that
\begin{eqnarray*}
z(\delta(S))&=& y(\delta(S))+|F|-2y(F)\\
&\ge&(1+4\alpha) |E_1| + (-1-\alpha) x(E_1) + (1+2\alpha) |E_2| + (-1+\alpha) x(E_2) \\
&&+ |E_3| + (-1+3\alpha) x(E_3) + \beta |E_3|\\
&\ge & 3\alpha (|E_1|+|E_2|+|E_3|) +\beta |E_3|\ge 3\alpha |F|,
\end{eqnarray*}
where the penultimate inequality follows from the fact that $x(E_k)\le |E_k|$ and $\alpha \leq \frac 13$, and the last inequality from the fact that $\alpha \le \frac 13$.
Hence, if we choose $\alpha = \frac 19$, then $z(\delta(S))\ge 1$.
\end{proof}

\end{proofof}

\section{Conjectures and Conclusions}
\label{sec:conc}

\iftoggle{abs}
{We conclude our paper with a conjecture.}
{
\begin{quote}
{\em
I conjecture that there is no [polynomial-time] algorithm for the traveling salesman problem.  My reasons are the same as for any mathematical conjecture: (1) It is a legitimate mathematical possibility, and (2) I do not know.
\begin{flushright}
--- Edmonds~\cite{Edmonds67}
\end{flushright}
}
\end{quote}

We conclude our paper with a conjecture.  We do so in the spirit of Jack Edmonds, quoted above; we do not know whether the conjecture is true or not, but we think that even a proof that this conjecture is false would be interesting.}
Our conjecture says that the integrality gap (or worst-case ratio) of the subtour LP is obtained for specific kinds of vertices of the subtour polytope; namely, ones in which the
subtour LP solution has no subtour constraint as part of the dual basis,
or, restated a different way, for costs $c$ such that an optimal subtour LP solution for $c$ is the same as an optimal fractional 2-matching for $c$.  Let us call such costs $c$ {\em fractional 2-matching costs} for the subtour LP.  Note that for such solutions of the subtour LP, the fractional 2-matching will have no cut edge.
\begin{conjecture} \label{conj:tsp}
The integrality gap for the subtour LP is attained for a fractional 2-matching cost for the subtour LP.
\end{conjecture}
We could make a similar conjecture for the ratio of the cost of the optimal 2-matching to the subtour LP, but by Theorem~\ref{thm:f2m-109} and Corollary~\ref{cor:boydcarr}, we already know that the conjecture is true.  However, its truth does not shed any light on the conjecture above.

In a companion paper, Qian et al.~\cite{QianSWvZ11} show that if an analogous conjecture for edge costs $c(i,j) \in \{1,2\}$ is true, then the integrality gap for 1,2-TSP is at most $\frac 76$.  They conjecture that the integrality gap for the 1,2-TSP is at most $\frac{10}9$; it is known that it can be no smaller than $\frac{10}9$.  It would be nice to show that if the analogous conjecture is true then the integrality gap for 1,2-TSP is at most $\frac{10}9$.

Interestingly, we appear to know almost nothing about the consequences of Conjecture~\ref{conj:tsp}.  Even for this very restricted set of cost functions, we do not know a better upper bound on the integrality gap of the subtour LP other than the bound of $\frac{3}{2}$.  Note that the lower bound of $\frac{4}{3}$ is attained for a fractional 2-matching cost.  It would be very interesting to prove that for such costs the integrality gap is indeed $\frac{4}{3}$.  Boyd and Carr~\cite{BoydC11} have shown this for some fractional 2-matching costs in which all the cycles of the fractional 2-matching have size 3; this result also follows from the technique of Theorem~\ref{thm:f2m-43-cutedge}, since the resulting graphical 2-matching is Eulerian if all cycles have size 3 and the fractional 2-matching has a single component (the graphical 2-matching may not be connected if there are cycles of size 5).

\subsection*{Acknowledgements}
We thank Sylvia Boyd for useful and encouraging discussions; we also thank her for giving us pointers on her various results.  Gyula Pap made some useful suggestions regarding the polyhedral formulation of graphical 2-matchings.

\bibliographystyle{abbrv}
\bibliography{subtour}

\appendix

\iftoggle{abs}{
\section{Figures} \label{sec:figs}

\begin{figure}[h]
\begin{center}
\subfloat[Type I.]{\includegraphics[width=.22\textwidth]{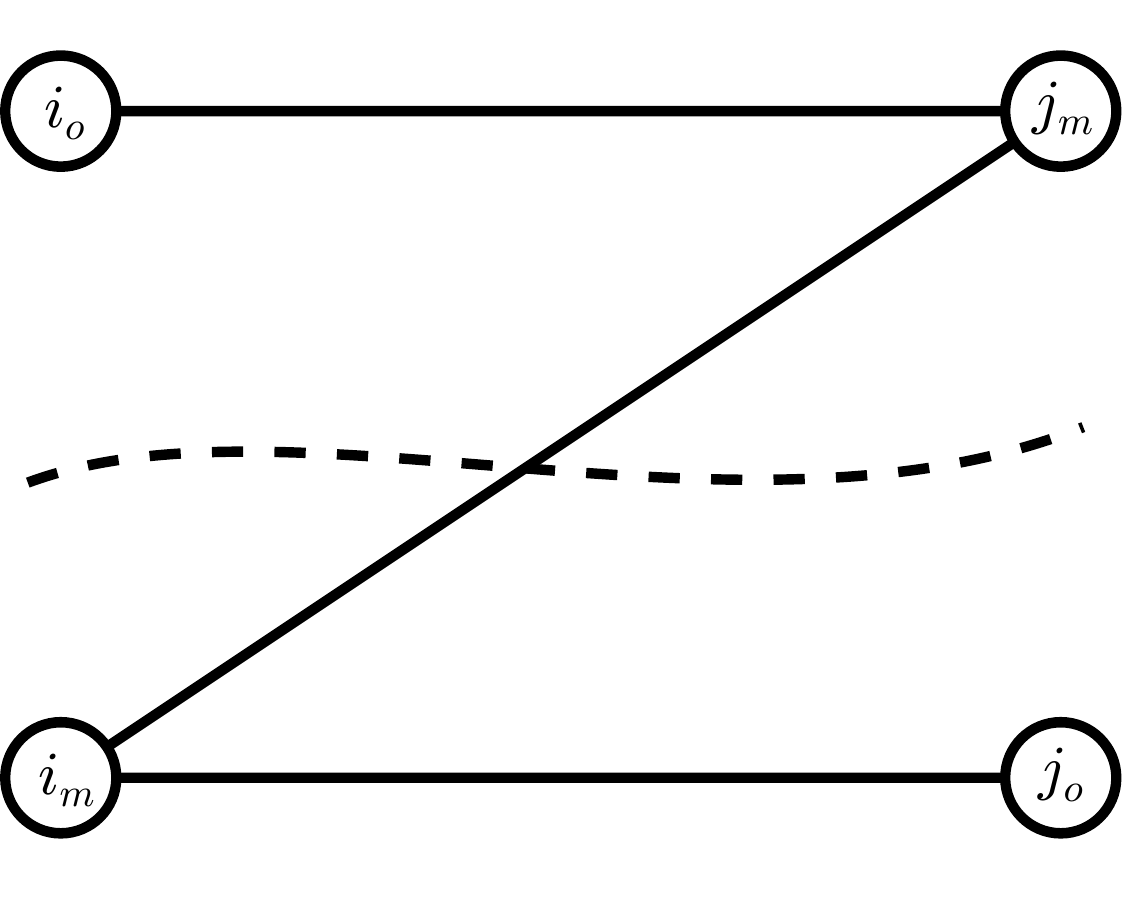}}\qquad
\subfloat[Type II.]{\includegraphics[width=.22\textwidth]{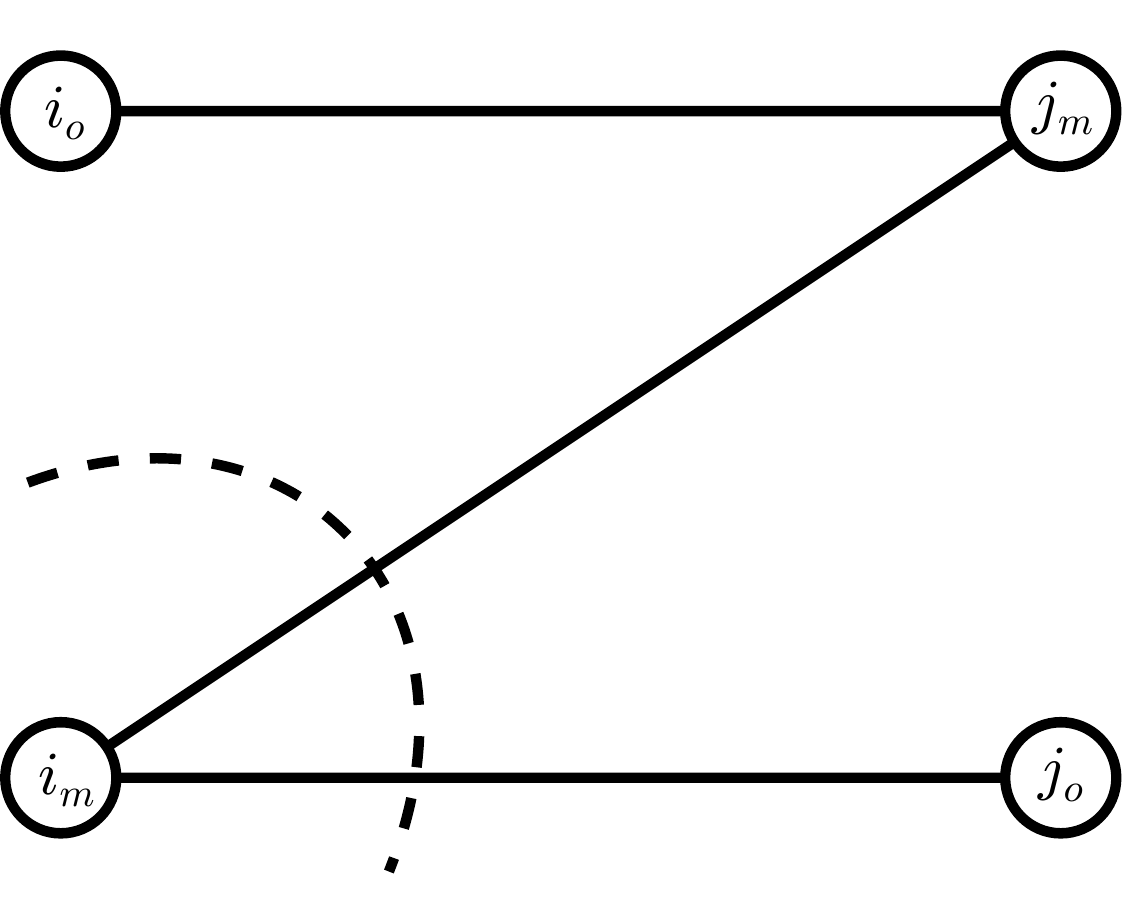}}\qquad
\subfloat[Type III.]{\includegraphics[width=.22\textwidth]{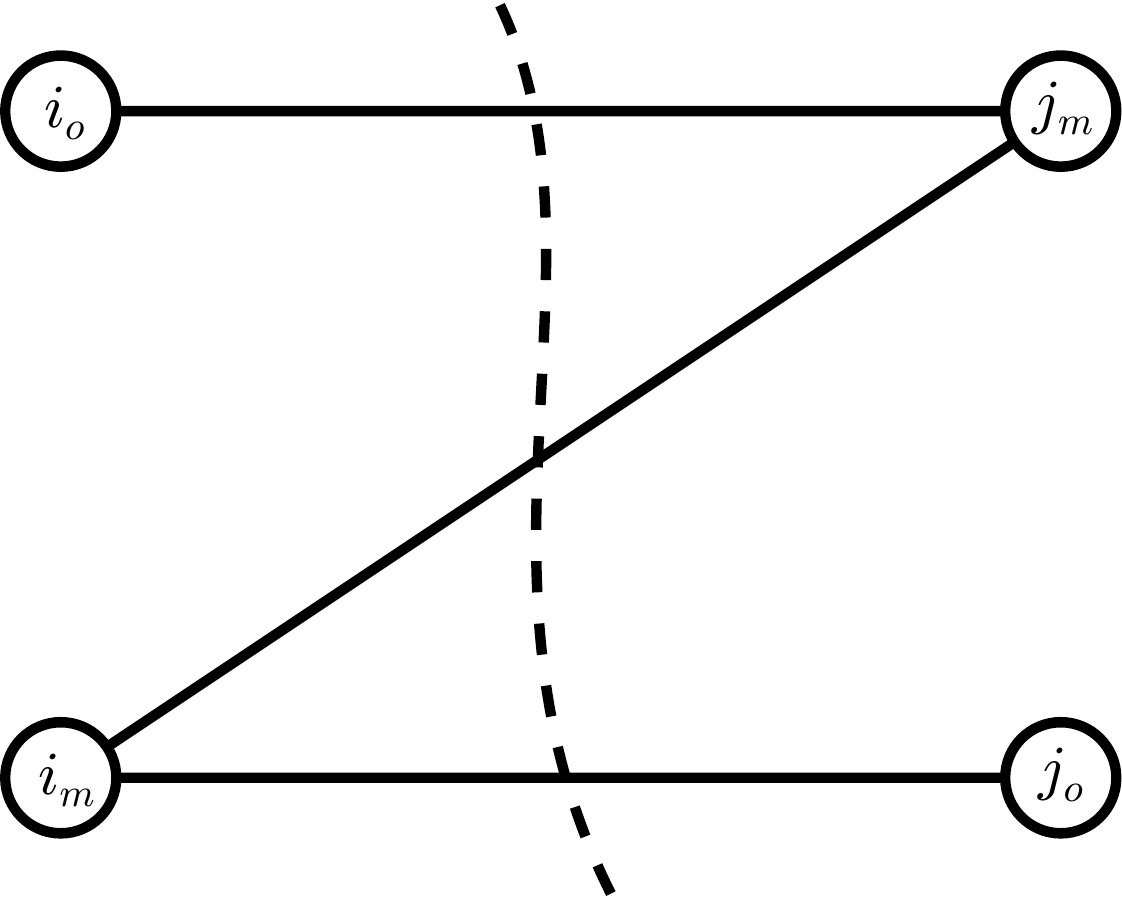}}\\
\subfloat[Type IV.]{\includegraphics[width=.22\textwidth]{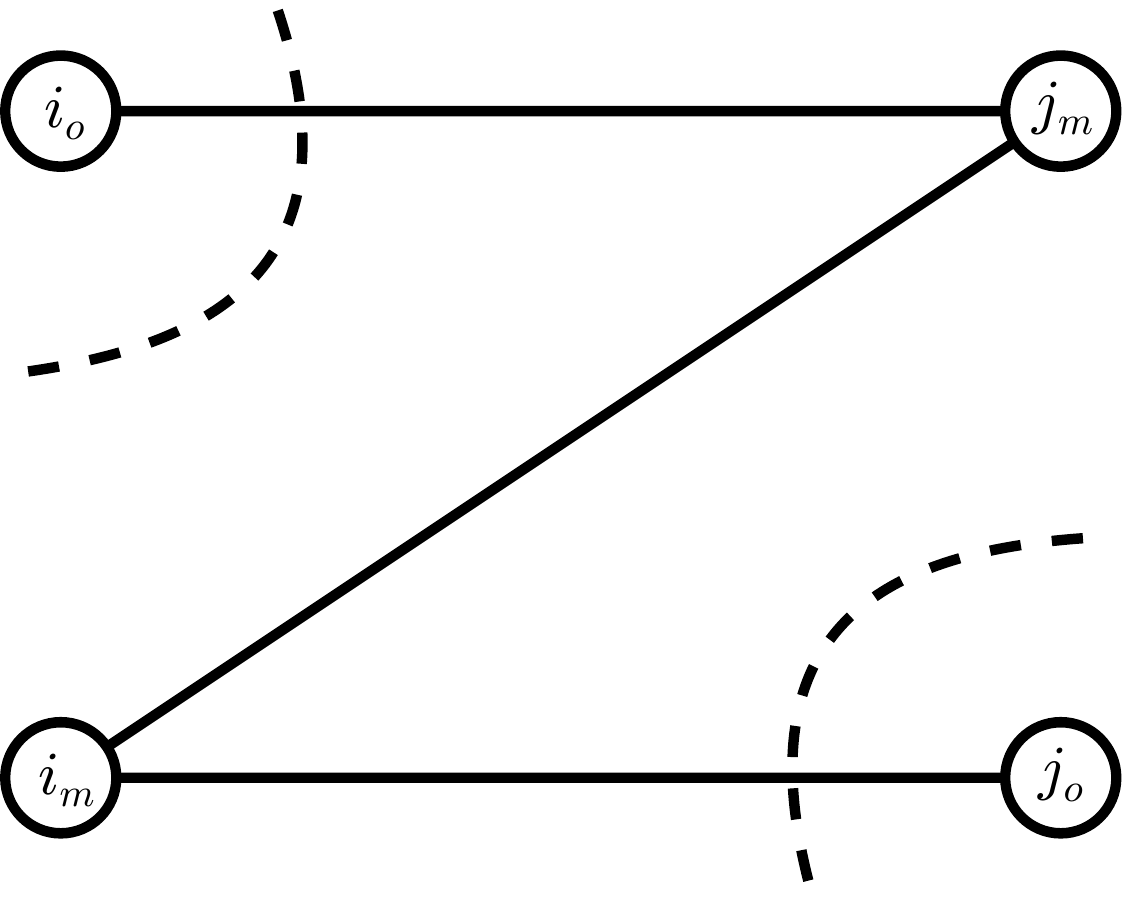}}\qquad
\subfloat[Type V.]{\includegraphics[width=.22\textwidth]{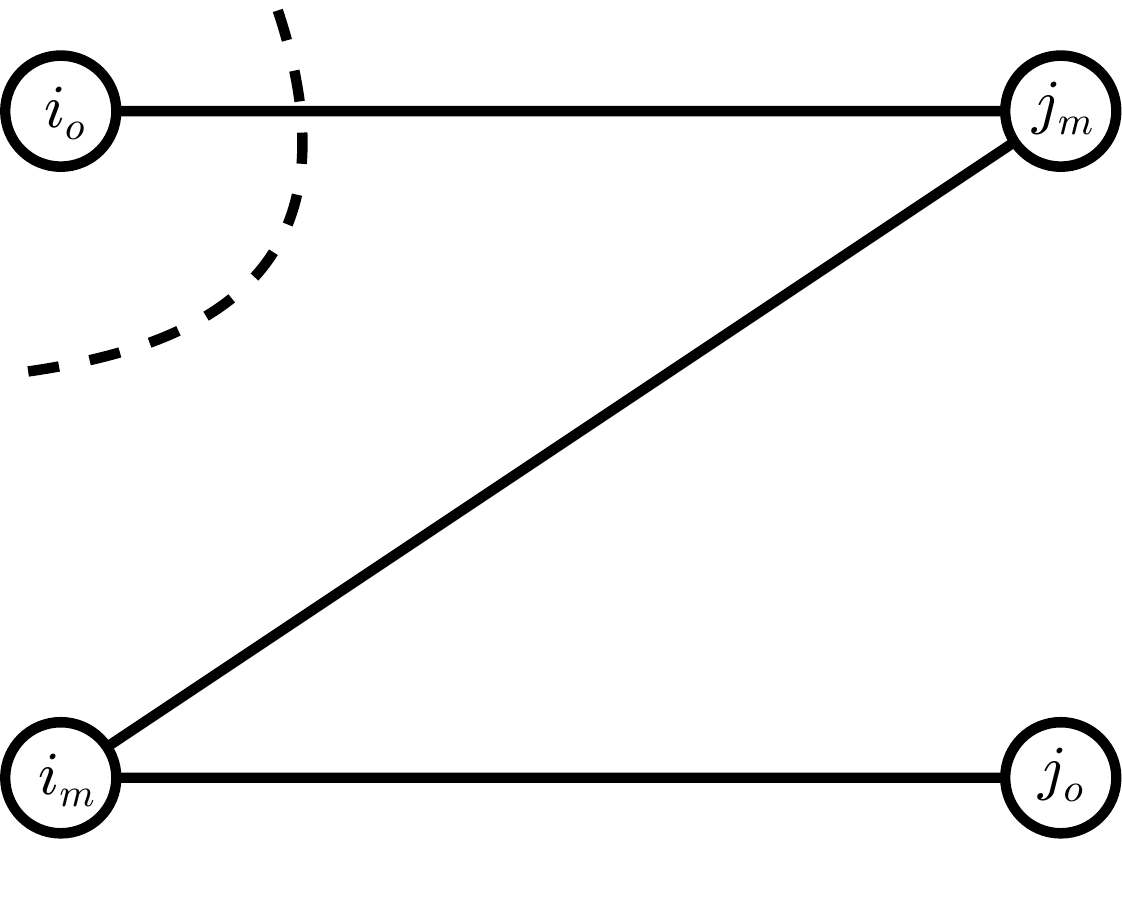}}\qquad
\end{center}
\caption{Illustrations of the five types of cuts of the edges in the reduction. The $y$-value on the top and bottom edge is $\alpha x(i,j)$ and the $y$-value on the middle edge is $(1-\alpha)x(i,j)$.}
\label{fig:types}
\end{figure}
}{}

\iftoggle{abs}{}{
\section{Polyhedral description of 2MO} \label{sec:2foproof}
We repeat Theorem~\ref{2fo} for sake of completeness.

\begin{theorem}
Let $G = ( V_\man \cup V_\opt, E )$ be a 2MO instance. The convex hull of integer 2MO solutions is given by the following polytope:
\lps
& & & & \sum_{ e \in \delta( i ) } x(e) = 2, & \forall i \in V_\man \numb{2fodegreeconsman2}\\
& & & & \sum_{ e \in \delta( i ) } x(e) \leq 2, & \forall i \in V_\opt \numb{2fodegreeconsopt2} \\
& & & & \sum_{ e \in \delta( S ) \setminus F } x(e) + \sum_{ e\in F } ( 1 - x(e) ) \geq 1, &\forall S \subseteq V, \,F \subseteq \delta( S ), \,F
\text{ matching}, \,|F| \text{ odd,} \numb{2foFcons2} \\
& & & & 0 \leq x(e) \leq 1, &\forall e\in E. \numb{2foboundscons2} \elps
\end{theorem}

\begin{proof}
The proof that we present here is similar to the proof of the polyhedral description of the 2-matching polytope (Theorem 30.8) in Schrijver~\cite{Schrijver-book}. We will first show that any 2MO solution is contained in the polytope, and next show that the extreme points of the polytope coincide with the 2MO solutions.

Constraints~(\ref{2fodegreeconsman2}), (\ref{2fodegreeconsopt2}) and (\ref{2foboundscons2}) obviously hold for a 2MO solution. To show that constraint~(\ref{2foFcons2}) is satisfied, we consider two cases:
(case 1) There is a $\bar e \in F$ with $x( \bar e ) = 0$. This makes the left hand side of constraint~(\ref{2foFcons2}) at least 1, since $x(e) \geq 0$ for all $e$.
(case 2) $x( e ) = 1$ for all $e\in F$. Since $|F|$ is odd, and each node is incident to an even number of edges in an 2MO solution, it follows that there has to be an edge in the solution in $\delta( S )$ that is not in $F$. So the constraint also holds in this case.

The polytope thus contains all 2MO solutions. We will now show that its extreme points coincide with 2MO solutions, by reducing 2MO instances to matching instances, for which perfect matchings correspond to 2MO solutions. We will show that any feasible point in the 2MO polytope corresponds to a feasible point in the perfect matching polytope. Because any point in the perfect matching polytope can be written as a convex combination of perfect matchings this implies that any point in the 2MO polytope can be written as a convex combination of 2MO solutions, and therefore all extreme points of the 2MO polytope correspond to 2MO solutions.

Before we consider the reduction to perfect matchings, we will first show that adding constraint~(\ref{2foFcons2}) for all $F \subseteq E$ of odd cardinality does not change the 2MO polytope. These additional constraints will be convenient when showing that a feasible point in the 2MO polytope is in the perfect matching polytope.

We prove this by induction on $|F|$. Consider $\bar S$ and $\bar F \subseteq \delta( \bar S )$ so that $F$ is not a matching, i.e. $| \bar F \cap \delta( i ) | \geq 2$ for some $i\in V$. We consider three cases.
\begin{itemize}
\item (Case 1) $| \bar F \cap \delta( i ) | \geq 3$.
Then
\begin{align*}
	\sum_{ e\in \delta( \bar S ) \setminus \bar F } x(e) + \sum_{ e\in \bar F } ( 1 - x(e) ) &\geq \sum_{ e\in \bar F } ( 1 - x(e) ) \geq \sum_{ e\in \bar F \cap \delta( i ) } ( 1 - x(e) ) \geq 3 - \sum_{ e\in \bar F \cap \delta( i ) } x(e) \\ &\geq 3 - \sum_{ e\in \delta( i ) } x(e) \geq 3 - 2 \geq 1.
\end{align*}
\item (Case 2) $| \bar F \cap \delta( i ) | = 2$ and $i \in \bar S$.  Let $F' = \bar F \setminus \delta( i )$  and let $S' = \bar S \setminus \{ i \}$. Then
\begin{align*}
\sum_{  e \in \delta( \bar S ) \setminus \bar F } x(e) & + \sum_{ e\in \bar F  } ( 1 - x(e) )\\
     &\geq \sum_{ e\in\delta( S' ) \setminus F' } x(e) - \sum_{e \in \delta( i ) } x(e) + \sum_{e \in \delta( i ) \cap \bar F } x(e) + \sum_{e \in F'} ( 1- x(e) ) + \sum_{ e\in \delta( i ) \cap \bar F } ( 1- x(e) ) \\
	&= \sum_{ e\in\delta( S' ) \setminus F' } x(e) + \sum_{ e\in F' }( 1- x(e) ) - \sum_{ e\in \delta( i ) } x(e) + 2.
\end{align*}
By induction and the degree bound for $i$, this quantity is at least 1.
\item (Case 3) $| \bar F \cap \delta( i ) | = 2$ and $i \not\in \bar S$. Let $F' = \bar F \setminus \delta( i )$ as in the previous case, but now let $S' = \bar S \cup \{ i \}$. Then the exact same string of inequalities as in the previous case holds.
\end{itemize}

We now use the usual reduction from 2-matchings to matchings (see Theorem 30.7 in Schrijver, the notation of which we will also follow): for each node $i$ in the 2MO, there will be two nodes in the matching instance: $i'$ and $i''$. For each edge $e = (i, j)$ in the 2MO instance, there will be two nodes and five edges in the matching instance: nodes $p_{ e, i }$ and $p_{ e, j }$, and edges $(i', p_{ e, i })$, $(i'', p_{ e, i })$, $( p_{ e, i }, p_{ e, j })$, $(j', p_{ e, j })$, and $(j'', p_{ e, j })$. The only difference between the reduction from 2-matchings to matchings, and the reduction from 2MO to matchings is that for optional nodes we also add an edge between nodes $i'$ and $i''$. An illustration of the reduction is given in Figure~\ref{fig:2foreduction}, where the part of the matching instance is given which corresponds to an edge between a mandatory node $i$, and an optional node $j$.

\begin{figure}
\begin{center}
\includegraphics[width=.42\textwidth]{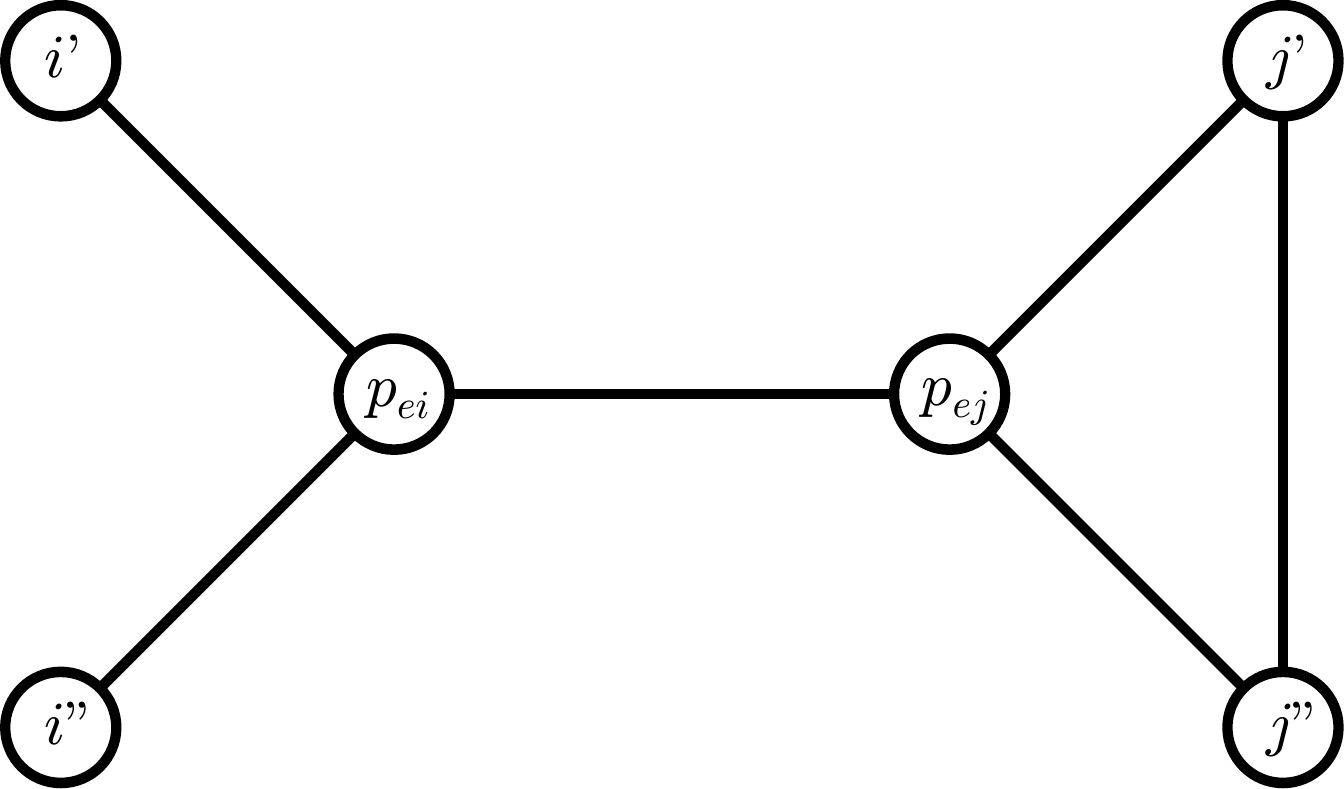}
\end{center}
\caption{Illustration of the reduction from 2MO to matchings. The part of the matching instance is drawn which corresponds to an edge between a mandatory node $i$, and an optional node $j$.}
\label{fig:2foreduction}
\end{figure}

Given a (fractional) solution $x$ to a 2MO instance, we define a solution $y$ to the corresponding matching instance as follows:
\begin{align*}
	y(i', p_{ e, i }) &= y( i'', p_{ e, i }) = \frac 12 x(e) \text{ and } \\
	y( p_{ e, i }, p_{ e, j }) &= 1 - x(e)
\end{align*}
for all $e = (i, j) \in E$, and
\[
	y(i', i'') = 1 - \frac 12 \sum_{ e\in\delta( i ) } x(e) \text{ for all } i \in V_\opt.
\]

We will now show that this solution is indeed in the perfect matching polytope given by the constraints~(\ref{matchingdegreecons}), (\ref{matchingScons}) and (\ref{matchingboundscons}) of the linear program $(M)$ in Section~\ref{sec:prelims} (where the variables are here called $y$ instead of $x$). For nodes $p_{e,i}$, the degree bound constraints~(\ref{matchingdegreecons}) follow directly from the definition of $y$ (there are three edges incident on $p_{e,i}$ with $y$-values $\frac 12 x(e)$, $\frac 12 x(e)$ and $1-x(e)$, which sum to $1$). For the other nodes, constraint~(\ref{matchingdegreecons}) follows directly from the degree bound constraints~(\ref{2fodegreeconsman2}) or~(\ref{2fodegreeconsopt2}) in the 2MO instance and the definition of $y$. Constraints~(\ref{matchingboundscons}) follow directly from constraints~(\ref{2foboundscons2}).

We will now prove that constraints~(\ref{matchingScons}) also hold for all subsets of nodes of odd cardinality in our reduction. Let $S'$ be such a subset. We consider four cases.
\begin{itemize}
\item (Case 1) $| \{ i', i'' \} \cap S' | = 1$ for some $i \in V$. Note that we have edges $(i', p_{ e, i })$ and $(i'', p_{ e, i })$ in the reduction both of which have $y$-value $\frac 12 x(e)$, and of which exactly one will be in $\delta( S' )$. Furthermore, $(i', i'')$ is in $\delta(S')$ if $i$ is in $V_\opt$. Therefore $\sum_{ e'\in \delta( S' ) } y(e') \geq \sum_{ e \in \delta( i ) } \frac 12 x(e) = 1$ if $i \in V_\man$ by the degree bound~(\ref{2fodegreeconsman2}). Similarly $\sum_{ e'\in \delta( S' ) } y(e') \geq \sum_{ e \in \delta( i ) } \frac 12 x(e)  + ( 1 - \frac 12 \sum_{ e \in \delta( i ) } x(e) ) = 1$ if $i \in V_\opt$ by the degree bound~(\ref{2fodegreeconsopt2}).
\item (Case 2) For some $e = (i, j) \in E$, $p_{e, i} \in S', p_{e, j} \not\in S'$ and $\{ i', i'' \} \cap S' = \emptyset$. Let $p = p_{e,i}$. Then $\sum_{ e' \in \delta( S' ) } y(e') \geq y( p, i' ) + y(p, i'') + y(p, p_{e,j}) = \frac 12 x(e) + \frac 12 x(e) + 1 - x(e) = 1$.
\item (Case 3) For some $e = (i, j)\in E$, $p_{e, i} \in S', p_{e, j} \not\in S'$ and $\{ j', j'' \} \subseteq S'$. Let $p = p_{e,j}$. Then $\sum_{ e' \in \delta( S' ) } y(e') \geq y(p, j') + y(p, j'') + y(p, p_{e,i}) = \frac 12 x(e) + \frac 12 x(e) + 1 - x(e) = 1$.
\item (Case 4) We may now assume that $S'$ is such that $| \{ i', i'' \} \cap S' |$ is even for all $i \in V$, and that $\{ i', i'' \} \in S'$ and $\{ j', j'' \} \cap S' = \emptyset$ if $p_{e, i} \in S'$ and $p_{e, j} \not\in S'$, because otherwise we are in one of the previous cases. Define $\bar S = \{ i \in V: i' \in S' \mbox{ and } i'' \in S' \}$ and $\bar F = \{ e = \{ i, j \} \in E: p_{ e, i } \in S' \mbox{ and } p_{ e, j } \not\in S'  \}$. Note that the previous argument implies that $\bar F \subseteq \delta( \bar S )$.

Consider $e = ( i, j ) \in \delta( \bar S )$ in the 2MO instance, and assume without loss of generality that $i \in \bar S$. By definition of $\bar S$, this means $\{ j', j'' \} \cap S' = \emptyset$.  We consider $e \in \bar F$ and $e \not\in\bar F$ separately.
First of all, assume $e \in \bar F$. Since we are not in the previous cases this means that $p_{ e, i } \in S'$ and $p_{ e, j } \not\in S'$. So for each such $e$ in the 2MO instance, we have $(p_{ e, i }, p_{ e, j }) \in \delta( S' )$ in the matching instance, with an $y$-value of $1-x(e)$.
Second, assume $e \not\in \bar F$. By definition of $\bar F$, we know that either $p_{e,i}$ and $p_{e,j}$ are both in $S'$, or both not in $S'$. So for each such $e$ in the 2MO instance, we have either $\{ (i', p_{ e,i }), (i'', p_{ e,i }) \} \subseteq \delta( S' )$ or $\{ ( j', p_{ e,j } ), ( j'', p_{ e,j } ) \} \subseteq \delta( S' )$ in the matching instance, each of which carry a total $y$-value of $x(e)$.

We thus get $\sum_{ e'\in\delta( S' ) } y(e') \geq \sum_{ e \in \delta(\bar S) \setminus \bar F } x(e) + \sum_{ e \in \bar F } ( 1 - x(e) )$.  We then note that $|\bar F|$ is equal to the number of nodes of the type $p_{e,i}$ in $S'$, which implies that the parity of $|\bar F|$ and $|S'|$ are always the same, as the other nodes in $S'$ appear in pairs.  Thus since $|S'|$ is odd, $|\bar F|$ is odd, and we have
 $\sum_{ e'\in\delta( S' ) } y(e') \geq \sum_{ e \in \delta(\bar S) \setminus \bar F } x(e) + \sum_{ e \in \bar F } ( 1 - x(e) ) \geq 1$ by the feasibility of $x$ for constraints (\ref{2foFcons2}).
\end{itemize}

We conclude the proof by noting that a perfect matching in the constructed instance corresponds to the 2MO solution consisting of all edges $e = ( i, j )$ for which $(p_{ e, i }, p_{ e, j })$ is {\em not} in the perfect matching solution.
\end{proof}}

\end{document}